\newif\ifital  \italfalse  
\newif\ifdraft  \draftfalse  
\ifital 
\documentclass[a4paper,11pt,twoside]{article}
\else
\documentclass[a4paper,11pt]{article}
\fi
\newlength{\augwidth}
\newlength{\augheight}
\usepackage{amsmath,amsxtra,amsfonts,latexsym,amscd,amssymb,amsthm,%
            amstext}
\usepackage{authblk}
\usepackage{xspace}
\usepackage{graphicx} 
\usepackage{vaucanson-g}
\newcommand{\JSLandscape}%
  {\ifital \thispagestyle{empty}\mbox{ }\clearpage%
      \addtocounter{page}{-1}\fi}
\usepackage{calc}
\newcounter{hours}\newcounter{minutes}
\newcommand\printtime%
  {\setcounter{hours}{\time/60}%
   \setcounter{minutes}{\time-\value{hours}*60}%
  \thehours\,h\,\,\theminutes}

\makeatletter
\def\@evenfoot{%
\makebox[0pt][l]{\footnotesize
\ifdraft\texttt{Paper in JCSC}\fi}  
\hfil -- {\normalsize \thepage} -- \hfil 
\makebox[0pt][r]{\footnotesize
\ifdraft\texttt{started on  17/05/21}\fi}  
}
\def\@oddfoot{%
\makebox[0pt][l]{\footnotesize
\ifdraft\texttt{Compiled \today}\fi}  
\hfil -- {\normalsize \thepage} -- \hfil 
\makebox[0pt][r]{\footnotesize
\ifdraft\texttt{at \printtime}\fi}  
}
\makeatother
\usepackage[mathlines,displaymath]{lineno}
\nolinenumbers 
\ifdraft
\linenumbers
\usepackage[notcite,notref]{showkeys}
\fi
\newcommand{\NEM}[1]%
   {\ifdraft%
    \marginpar[\begin{flushright}%
               {\sl {\scriptsize #1}}%
               \end{flushright}]%
              {\begin{flushleft}%
               {\sl {\scriptsize #1}}%
               \end{flushleft}}%
	\fi}%
\newcommand{\CMT}[1]%
   {\ifdraft%
    \par
    \noindent
    XXXXXX\PushLine  
{\small\textsl{Commentaire trop grand pour rentrer dans la marge...}} 
    \PushLine XXXXXX
    \par
    \textsl{#1}
    \par
    \noindent
    XXXXX\PushLine 
{\small\textsl{...mais ce n'est pas une raison pour ne rien \'ecrire}} 
    \PushLine XXXXXX\\
    \fi}%
\newcommand{\cmt}[1]%
   {\ifdraft%
    \par
    \noindent
	XXXXXX\PushLine \textsc{#1}
	\PushLine XXXXXX\\
 	\fi}%
\newcommand{\BibDir}{./bibinputs/}%
\theoremstyle{plain}
\newtheorem{theorem}{Theorem}
\newtheorem{corollary}[theorem]{Corollary}
\newtheorem{lemma}[theorem]{Lemma}
\newtheorem{proposition}[theorem]{Proposition}

\newtheorem{claim}{Claim}
\theoremstyle{definition}

\newtheorem{definition}[theorem]{Definition}
\newtheorem{example}[theorem]{Example}
\makeatletter
\newenvironment{proofss}[1][\proofname]{\par
  \normalfont \topsep6\p@\@plus6\p@\relax
  \trivlist
  \item[\hskip\labelsep
        \itshape
    #1\@addpunct{.}]\ignorespaces
}{%
  \endtrivlist\@endpefalse
}
\makeatother
\newcommand{\clmcl}[1]{Claim~\ref{cl.#1}}

\newcommand{\equat}[1]{Equation~(\ref{q.#1})}
\newcommand{\equnm}[1]{(\ref{q.#1})\xspace}

\newcommand{\figur}[1]{Figure~\ref{f.#1}}
\newcommand{\lemme}[1]{Lemma~\ref{l.#1}}
\newcommand{\propo}[1]{Proposition~\ref{p.#1}}

\newcommand{\theor}[1]{Theorem~\ref{t.#1}}
\newcommand{\secti}[1]{Section~\ref{s.#1}}
\newcounter{theor-tmp}

\newcounter{exa-exp}%
\newcommand{\idex}{1}%
\newcommand{\LatinLocution}[1]{{\itshape #1}\xspace}
\newcommand{\cf}{\LatinLocution{cf.}}

\newcommand{\etc}{\LatinLocution{etc.}} 

\newcommand{\ie}{{that is, }}

\newcommand{\e}{\text{\quad}}                 
\newcommand{\ee}{\text{\qquad}}               
\newcommand{\eee}{\text{\qquad \qquad}} 
\newcommand{\Defi}[2]%
    {\left\{#1\xmd\xmd\middle|\xmd\xmd#2\right\}}
\newlength{\retraita}
\newcommand{\pointn}{\noindent \makebox[1.2em]{$\bullet$}\ignorespaces}
\newcommand{\fa}{\forall}
\newsavebox{\InterSymbolSpace}
\savebox{\InterSymbolSpace}{\hspace{0.125em}}
\newsavebox{\SideFormulaSpace}
\savebox{\SideFormulaSpace}{\hspace{0.2em}}
\savebox{\SideFormulaSpace}{\hspace{0.2em}}
\newcommand{\msp}{\usebox{\SideFormulaSpace}} 
\newcommand{\xmd}{\usebox{\InterSymbolSpace}} 
\newcommand{\eqpnt}{\makebox[0pt][l]{\: .}}
\newcommand{\eqvrg}{\makebox[0pt][l]{\: ,}}
\newcommand{\EqVrgInt}{\: , \e }
\newcommand{\quantvrg}{\, , \;}

\newcommand{\EqPnt}{\: .}

\newcommand{\quantsp}{\ee }
\newcommand{\quantsmsp}{\e }
\newcommand{\bk}{\setminus}
\newcommand{\Ac}{\mathcal{A}}
\newcommand{\Bc}{\mathcal{B}}

\newcommand{\Dc}{\mathcal{D}}
\newcommand{\Fc}{\mathcal{F}}

\newcommand{\Ic}{\mathcal{I}}
\newcommand{\Jc}{\mathcal{J}}

\newcommand{\Sc}{\mathcal{S}}
\newcommand{\Tc}{\mathcal{T}}
\newcommand{\Uc}{\mathcal{U}}

\newcommand{\mathjsu}[1]{\mathsf{#1}}
\newcommand{\redmatu}[1]{\scalebox{0.84}{#1}}

\newcommand{\matriceuu}[1]%
    {\begin{pmatrix} #1 \end{pmatrix}}
\newcommand{\matricedd}[4]%
    {\begin{pmatrix} #1 & #2 \\ #3 & #4 \end{pmatrix}}
\newcommand{\vecteurd}[2]%
    {\begin{pmatrix} #1 \\ #2 \end{pmatrix}}
\newcommand{\ligned}[2]%
    {\begin{pmatrix} #1 & #2 \end{pmatrix}}
\newcommand{\matricett}[9]%
    {\begin{pmatrix}  #1 & #2 & #3 \\
                      #4 & #5 & #6 \\
                      #7 & #8 & #9 \end{pmatrix}}
\newcommand{\vecteurt}[3]%
    {\begin{pmatrix} #1 \\ #2 \\ #3 \end{pmatrix}}
\newcommand{\lignet}[3]%
    {\begin{pmatrix} #1 & #2 & #3 \end{pmatrix}}

\newlength{\jsWidthCol}
\newlength{\blocinterligne}
\newlength{\blocinterligned}
\newlength{\temparraycolsep}
\newlength{\longueurbloc}
\newlength{\hauteurbloc}
\newlength{\centragebloc}
\newlength{\longueurblc}
\newlength{\hauteurblc}
\newlength{\centrageblc}
\newcommand{\blocligne}[1]%
    {\framebox[\longueurbloc]{$#1$}}
\newcommand{\blocmatrice}[1]%
    {\framebox[\longueurbloc]{\rule[\centragebloc]{0mm}{\hauteurbloc}$#1$}}
\newcommand{\blocvecteur}[1]%
    {\framebox{\rule[\centragebloc]{0mm}{\hauteurbloc}$#1$}}
\newcommand{\blcligne}[1]%
    {\framebox[\longueurblc]{$#1$}}
\newcommand{\blcmatrice}[1]%
    {\framebox[\longueurblc]{\rule[\centrageblc]{0mm}{\hauteurblc}$#1$}}
\newcommand{\blcvecteur}[1]%
    {\framebox{\rule[\centrageblc]{0mm}{\hauteurblc}$#1$}}
\newcommand{\matriceddblvs}[4]
   {\setlength{\temparraycolsep}{\arraycolsep}%
    \setlength{\arraycolsep}{1.3pt}%
        \left (%
    \begin{array}{cc}%
                #1  & \blcligne{#2} \\
            \blcvecteur{#3} & \blcmatrice{#4}
        \end{array}%
        \right )%
    \setlength{\arraycolsep}{\temparraycolsep}%
   }%
\newcommand{\vecteurdblvs}[2]%
   {\setlength{\temparraycolsep}{\arraycolsep}%
    \setlength{\arraycolsep}{1.5pt}%
        \left (%
    \begin{array}{c}%
                #1  \\
            \blcvecteur{#2}
        \end{array}%
        \right )%
    \setlength{\arraycolsep}{\temparraycolsep}%
   }%
\newcommand{\lignedblvs}[2]%
   {\setlength{\temparraycolsep}{\arraycolsep}%
    \setlength{\arraycolsep}{1.5pt}%
        \left (%
    \begin{array}{cc}%
                #1  & \blcligne{#2}
        \end{array}%
        \right )%
    \setlength{\arraycolsep}{\temparraycolsep}%
   }%
\newcommand{\matricettblvs}[9]
   {\setlength{\temparraycolsep}{\arraycolsep}%
    \setlength{\arraycolsep}{1.5pt}%
        \left (%
    \begin{array}{ccc}%
                #1  & \blcligne{#2} & #3\\
            \blcvecteur{#4} & \blcmatrice{#5} & \blcvecteur{#6}\\
                #7  & \blcligne{#8} & #9\\
        \end{array}%
        \right )%
    \setlength{\arraycolsep}{\temparraycolsep}%
   }%
\newcommand{\vecteurtblvs}[3]%
   {\setlength{\temparraycolsep}{\arraycolsep}%
    \setlength{\arraycolsep}{1.5pt}%
        \left (%
    \begin{array}{c}%
                #1  \\
            \blcvecteur{#2}\\
                #3
        \end{array}%
        \right )%
    \setlength{\arraycolsep}{\temparraycolsep}%
   }%
\newcommand{\lignetblvs}[3]%
   {\setlength{\temparraycolsep}{\arraycolsep}%
    \setlength{\arraycolsep}{1.5pt}%
        \left (%
    \begin{array}{ccc}%
                #1  & \blcligne{#2} & #3
        \end{array}%
        \right )%
    \setlength{\arraycolsep}{\temparraycolsep}%
   }%
\newcommand{\matricettblblvs}[9]
   {\setlength{\temparraycolsep}{\arraycolsep}%
    \setlength{\arraycolsep}{1.5pt}%
        \left (%
    \begin{array}{ccc}%
                #1  & \blcligne{#2} & \blcligne{#3}\\
            \blcvecteur{#4} & \blcmatrice{#5} & \blcmatrice{#6}\\
                \blcvecteur{#7}  & \blcmatrice{#8} & \blcmatrice{#9}\\
        \end{array}%
        \right )%
    \setlength{\arraycolsep}{\temparraycolsep}%
   }%
\newcommand{\vecteurtblblvs}[3]%
   {\setlength{\temparraycolsep}{\arraycolsep}%
    \setlength{\arraycolsep}{1.5pt}%
        \left (%
    \begin{array}{c}%
                #1  \\
            \blcvecteur{#2}\\
                \blcvecteur{#3}
        \end{array}%
        \right )%
    \setlength{\arraycolsep}{\temparraycolsep}%
   }%
\newcommand{\lignetblblvs}[3]%
   {\setlength{\temparraycolsep}{\arraycolsep}%
    \setlength{\arraycolsep}{1.5pt}%
        \left (%
    \begin{array}{ccc}%
                #1  & \blcligne{#2} & \blcligne{#3}
        \end{array}%
        \right )%
    \setlength{\arraycolsep}{\temparraycolsep}%
   }%
\newcommand{\jsAutUn}[1]%
   {\mbox{$\left\langle \thinspace #1 \thinspace \right\rangle $}}
\newcommand{\aut}[1]{\jsAutUn{#1}} 
\newcommand{\autiet}{\jsAutUn{I, E ,T}}
\newcommand{\autjfu}{\jsAutUn{J, F ,U}}
\newlength{\vbh}\newlength{\vbd}\newlength{\vbt}%
\newcommand{\CompAut}[2]%
    {%
     \settodepth{\vbd}{\mbox{$\displaystyle{#1#2}$}}%
     \settoheight{\vbh}{\mbox{$\displaystyle{#1#2}$}}%
     \setlength{\vbt}{\vbh}\addtolength{\vbt}{\vbd}%
     {}%
     \psline[linewidth=0.8pt]{c-c}(0,-.65\vbd)(0,.9\vbh)%
     \hspace*{.15em}%
     {#1}%
     \hspace*{.15em}%
     \psline[linewidth=0.8pt]{c-c}(0,-.65\vbd)(0,.9\vbh)%
     }%
\newcommand{\CompAuto}[1]{{\CompAut{#1}{\strut}}}
\newcommand{\EoP}{\hbox{}\hfill\qedsymbol\hbox{}}%
\newcommand{\STR}{\begin{center} \Large * \end{center}}
\DeclareMathOperator{\crd}{Card}
\newcommand{\jsCard}[1]{\crd\!\left(#1\right)}
\newcommand{\RatExp}[2]{{#1\mathjsu{RatE}\,#2}}
\newcommand{\KRatEM}{\RatExp{\K}{M}}
\newcommand{\x}{\! \times \!}
\newcommand{\jsStar}[1]{{{#1}^{*}}}
\newcommand{\Ae}{\jsStar{A}}
\newcommand{\Be}{\jsStar{B}}
\newcommand{\PosiElmt}[1]{#1_{\bullet}}
\newcommand{\MPos}{\PosiElmt{M}}

\newcommand{\C}{\mathbb{C}}
\newcommand{\K}{\mathbb{K}}
\newcommand{\Kmbb}{\mathbb{K}}
\newcommand{\N}{\mathbb{N}}
\newcommand{\Q}{\mathbb{Q}}
\newcommand{\R}{\mathbb{R}}
\newcommand{\Z}{\mathbb{Z}}
\newcommand{\zed}{\mathsf{0}}
\newcommand{\und}{\mathsf{1}}
\newcommand{\Ed}{\mathsf{E}}
\newcommand{\Fd}{\mathsf{F}}
\newcommand{\Gd}{\mathsf{G}}
\newcommand{\Hd}{\mathsf{H}}
\newcommand{\Kd}{\mathsf{K}}
\newcommand{\Ld}{\mathsf{L}}
\newcommand{\TrmCnst}{{\mathsf{c}}}
\newcommand{\TermCnst}[1]{{\operatorname{\TrmCnst}\!\left(#1\right)}}
\newcommand{\TermCst}[1]{\TermCnst{#1}} 
\newcommand{\PartProp}[1]{#1_{\mathsf{p}}}
\newcommand{\spr}{\PartProp{s}}
\newcommand{\ETAze}[1]{0_{#1}}
\newcommand{\ETAun}[1]{1_{#1}}
\newcommand{\zeK}{\ETAze{\Kmbb}}
\newcommand{\unK}{\ETAun{\Kmbb}}
\newcommand{\unM}{\ETAun{M}}
\newcommand{\Pfrak}{\mathfrak{P}}
\newcommand{\jsPart}[1]{{\operatorname{\Pfrak}\left(#1\right)}}
\newcommand{\bra}[1]{\hbox{}\langle#1\rangle}%
\newcommand{\SerSAnMon}[2]%
    {#1 \langle \! \langle  #2  \rangle \! \rangle }
\newcommand{\KM}{\SerSAnMon{\K}{M}}
\newcommand{\KAe}{\SerSAnMon{\K}{\Ae }}
\newcommand{\KA}{\KAe} 
\newcommand{\Rat}{\mathrm{Rat}\,}
\newcommand{\KRA}{\RatExp{\K}{\Ae}}
\newcommand{\plusopr}{\mathbin{\mathjsu{+}}}
\newcommand{\prodopr}{\mathbin{\mathjsu{\cdot}}}

\newcommand{\CompExpr}[1]{\CompAuto{#1}}
\newcommand{\matmul}{\mathbin{\cdot}}
\newcommand{\autplus}{+}
\newcommand{\autprod}{\matmul}
\newcommand{\autstarsymb}{*}

\newcommand{\plusexp}{\mathsf{+}}
\newcommand{\prodexp}{\cdot}
\newcommand{\plusK}{\oplus}

\renewcommand{\KRatEM}{\RatExp{\K}{M}}
\newcommand{\KRatM}{\K\Rat M}
\newcommand{\ConjAuto}[1]{\overset{#1}{\Longrightarrow}}
\newcommand{\mmul}{\xmd}
\DeclareMathOperator{\LttLng}{\ell}
\newcommand{\LittLeng}[1]{\LttLng_{#1}}

\newcommand{\amlg}[1]{X_{#1}}%
\newcommand{\slct}[1]{Y_{#1}}%
\newcommand{\xphi}{\amlg{\varphi}}%
\newcommand{\yphi}{\slct{\varphi}}%
\newcommand{\Stan}[1]{\Sc_{#1}}
\newcommand{\StanE}{\Stan{\Ed}}

\newcommand{\Term}[1]{\Tc_{#1}}
\newcommand{\TermE}{\Term{\Ed}}
\newcommand{\TermF}{\Term{\Fd}}
\newcommand{\TermG}{\Term{\Gd}}
\newcommand{\Anti}[1]{\Ac_{#1}}
\newcommand{\AntiE}{\Anti{\Ed}}

\DeclareMathOperator{\Ntl}{\Ic}
\newcommand{\Nitl}[1]{\Ntl\!\left(#1\right)} 
\newcommand{\NitlE}{\Nitl{\Ed}}
\newcommand{\NitlF}{\Nitl{\Fd}}
\newcommand{\NitlG}{\Nitl{\Gd}}
\DeclareMathOperator{\Dff}{\mathrm{d}}%
\newcommand{\Diff}[1]{\Dff#1}
\newcommand{\DiffE}{\Diff{\Ed}}
\newcommand{\derv}[2]{\bra{\Nitl{#1},#2}}
\newcommand{\dervEa}{\derv{\Ed}{a}}

\newcommand{\StanAuto}[4]%
   {\aut{\redmatu{\ligned{1}{0}},%
         \redmatu{\matricedd{0}{#1}{0}{#2}},%
         \redmatu{\vecteurd{#3}{#4}}}}
\newcommand{\StanBloc}[4]%
   {\aut{\!\redmatu{\lignedblvs{1}{0}},%
         \redmatu{\matriceddblvs{0}{#1}{0}{#2}},%
         \redmatu{\vecteurdblvs{#3}{#4}}\!}}
\newcommand{\StanAutoSP}[9]%
   {\aut{\redmatu{\lignet{1}{0}{0}},%
         \redmatu{\matricett{0}{#1}{#2}%
                            {0}{#3}{#4}%
                            {0}{#5}{#6}},%
         \redmatu{\vecteurt{#7}{#8}{#9}}}}

\DeclareMathOperator{\Intl}{\Jc}
\newcommand{\Init}[1]{\Intl\!\left(#1\right)}
\DeclareMathOperator{\Mtrx}{\Fc}
\newcommand{\Matr}[1]{\Mtrx\!\left(#1\right)}
\DeclareMathOperator{\Fnl}{\Uc}
\newcommand{\Ermi}[1]{\Fnl\!\left(#1\right)}
\newcommand{\BigAuto}[1]%
  {\FixStateDiameter{2.8cm}%
   \scalebox{1.272 1}{\State{(0,0)}{XQ}}%
   \ChgStateLabelScale{2}%
   \VCPutStateLabel{(0,0)}{#1}%
   \RstStateLabelScale\MediumState}
\newcommand{\ExpDer}[2][a]%
    {\operatorname{\frac{\partial}{\partial \mbox{$#1$}}}#2}
\newcommand{\ExpDerP}[2][a]%
    {\operatorname{\frac{\partial}{\partial\mbox{$#1$}}}\left(#2\right)}
\newcommand{\ExpDerE}{\ExpDer{\Ed}}%
\newcommand{\ExpDerF}{\ExpDer{\Fd}}%
\newcommand{\ExpDerG}{\ExpDer{\Gd}}%
\newcommand{\ExpDerr}[2][a]%
    {\operatorname{\frac{\partial_{\mathrm{R}}}{\partial \mbox{$#1$}}}#2}
\newcommand{\ExpDerB}[2][a]%
   {\operatorname{\frac{\partial_\mathsf{b}}{\partial \mbox{$#1$}}}#2}
\newcommand{\ExpDerBP}[2][a]%
   {\operatorname{\frac{\partial_\mathsf{b}}{\partial \mbox{$#1$}}}\left(#2\right)}
\newcommand{\TerScale}{0.9}
\newcommand{\DTer}{\mathop{\scalebox{\TerScale}{\mathrm{D}}}}
\newcommand{\TDTer}{\scalebox{\TerScale}{\mathrm{TD}}}
\newcommand{\DerTer}[1]{{\DTer}\!\left(#1\right)}
\newcommand{\DerTerE}{\DerTer{\Ed}}
\newcommand{\DerTerF}{\DerTer{\Fd}}
\newcommand{\DerTerG}{\DerTer{\Gd}}

\newcommand{\TruDerTer}[1]{{\TDTer}\left(#1\right)}

\title{%
\vspace*{-2cm}
\textbf{Derived Terms without Derivation}\footnote{%
Published in \textit{Journal of Computer Science and Cybernetics},
Vol.~37, No.\,3 (2021), Special issue dedicated to the memory of 
Professor Phan Dinh Dieu, pp.\,201--221.}\\%
\smallskip
{\Large{%
\textbf{A shifted perspective on the derived-term automaton}}}
}

\author[1]{\textit{Sylvain Lombardy}}
\author[2]{\textit{Jacques Sakarovitch}}
\affil[1]{LaBRI, Bordeaux INP -- Bordeaux University -- CNRS}
\affil[2]{IRIF, CNRS -- Paris University and Telecom Paris, IPP}
\date{}

\begin{document}
\JSLandscape
\maketitle
\begin{abstract}
We present here a construction for the derived term automaton (aka 
partial derivative, or Antimirov, automaton) of a rational (or 
regular) expression based on a sole induction on the depth of the 
expression and without making reference to an operation of derivation 
of the expression.
It is particularly well-suited to the case of weighted rational 
expressions and the case of expressions over non free monoids.


\end{abstract}


\section{{Introduction}}%
\label{s.int}%

In this paper, we address once again the laboured problem of the 
transformation of a rational (regular) expression into a finite 
automaton that accepts the language, or the series, denoted by the 
expression.

In the Handbook of Automata Theory that appeared recently 
\cite{Pin21Ed}, we have given a survey on the many aspects of the 
transformation of an automaton into an expression and vice-versa, 
together with a comprehensive bibliography \cite{Saka21}.
We explain in this chapter that the equivalence between automata and
expressions may be generalised from `classical' automata and
expressions to \emph{weighted} automata and \emph{weighted}
expressions, to automata and expressions over monoids that are not
necessarily free monoids, and even to \emph{weighted} automata and
\emph{weighted} expressions over monoids that are not necessarily free
monoids (but still \emph{graded}).
This generalisation makes on one hand-side the relationship between
automata and expressions tighter and leads on the other hand-side to
`split' Kleene Theorem into two parts: the first one is the
correspondence between automata and expressions, and the second the
equality between the family of rational (or regular) languages or
series and the family of recognisable languages or series, an equality
which holds in the case of languages or series over free monoids only.

As we describe in that survey~\cite{Saka21}, there are two main
methods for computing an automaton from an expression which yield two
distinct, even though related, automata: the \emph{position automaton}
and the \emph{derived-term automaton} of the expression.

The first method can be credited to Glushkov \cite{Glus61}.
It associates with an expression of litteral length~$n$ a 
(non-deterministic) automaton with~$n+1$ states --- often called the 
\emph{Glushkov} or \emph{position} automaton of the expression.
As we recall in \secti{sta-aut} below the position automaton may be 
inductively defined by means of operations on automata of a certain 
kind that we call \emph{standard automata}.
The definition and computation of the position automaton readily
generalises to \emph{weighted expressions} \cite{CaroFlou03} and, even
more easily as there is nothing to change, to expressions over
\emph{non free monoids}.

It takes some more lines to sketch the second method.
It is well-known that a language (subset of a free monoid) is 
accepted by a finite automaton if and only if it has a \emph{finite 
number} of (left) quotients.
The starting point of the second method is the idea, due to 
Brzozowski, to lift this property of recognisable languages at the 
symbolic level of rational (or regular) expressions.
In \cite{Brzo64}, Brzozowski defined the \emph{derivatives} of a 
rational expression and turned them into the states of a 
deterministic automaton that recognises the language denoted by the 
expression.
He then showed that modulo the axioms of associativity, commutativity,
and idempotency of the addition (on the set of languages), the
ACI-properties, the set of derivatives of an expression is finite.
Under this form, it is clear that this `derivative' method is
essentially different from the first one as it cannot be generalised
to weighted rational expressions since weighted finite automata cannot
be determinised nor to expressions over non free monoids since 
subsets accepted by finite automata over such monoids have not 
necessarily a finite number of (left) quotients.

Thirty years later, Antimirov made another fundamental contribution 
to this theory and proposed a new derivation process~\cite{Anti96}.
Antimirov's derivation breaks Brzozowski's 
derivatives into parts --- hence the name `partial derivatives' 
given to these parts, a terminology we find unfortunate and we call 
them \emph{derived terms}.
As before, derived terms are turned into the states of a finite
automaton which we call \emph{derived-term automaton} and which
accepts the language denoted by the expression.
This construction has several outcomes.
The number of derived terms 
of an expression is not only finite but also `small', smaller than, 
or equal to, the litteral length of the expression.
Derived terms are defined without the usage of ACI-properties, which 
makes them easier to compute than the derivatives.

Finally, a link between the two methods was established somewhat later
by Champarnaud and Ziadi in~\cite{ChamZiad02}, and the derived-term 
automaton of an expression~$\Ed$ was shown to be a 
\emph{morphic image}\footnote{%
   Usually, one says that an automaton~$\Ac$ is a \emph{quotient} of 
   an automaton~$\Bc$ if there exists a morphism from~$\Bc$ 
   onto~$\Ac$, that is, if~$\Ac$ is a morphic image of~$\Bc$.
   In this introduction, we prefer this latter terminology as it does 
   not collide with the (left) \emph{quotient} of a language, or
   of a series (by a word).
   }
of the position automaton of~$\Ed$, a result which we refer to as the 
\emph{morphism theorem} in the sequel.

In~\cite{LombSaka05a}, we have extended the construction of the 
derived-term automaton to \emph{weighted expressions}.
Of course, the relationship with `derivatives' has disappeared in 
this generalisation, but the link the derivation of an expression and 
the \emph{quotient} of series is as strong as in the Boolean case.
The `weighted version' of the characterisation of recognisability with
the quotients is due to Jacob in full generality and reads as follows:
\emph{a series is recognisable if and only if it belongs to a finitely
generated submodule stable by quotient} (see~\cite{BersReut11}
or~\cite{Saka09} for instance).
And the derived terms of a weighted expression are a set of generators
of a module that contains the series denoted by the expression and
that is stable by quotient.
The construction of the same automaton has also been given by Rutten
as a byproduct of his theory of conduction on series which puts the
quotient operation on series at the first place~\cite{Rutt03}.

We also showed, in the same paper, that the morphism theorem quoted
above could be generalised to the weighted case, and, with the
adequate generalisation of the notion of morphism to weighted
automata, that the derived-term automaton of a weighted
expression~$\Ed$ is a \emph{morphic image} of the position automaton
of~$\Ed$.

It must be noted however that a difficulty arose in the proof of this 
last result.
In the derivation process and, if the weight semiring contains such
elements, some terms may vanish from the set of derived terms by the
interplay of `positive' and `negative' coefficients.
In such cases, they will be `missing' and the derived-term automaton
will not be a morphic image of the position automaton but only a
sub-automaton of a morphic image of the position automaton.
Instead of contenting ourselves with this weaker statement, we proved 
that the definition of the derived terms could be decorrelated from 
the derivation itself and obtained by induction on the expression and 
that the `morphism theorem' would then hold in full generality.
The proofs of the various properties of the derived-term automaton 
however relied on the connection with the derivation of the 
expression and the quotients of the series.

\STR
\enlargethispage{1ex}%

This long presentation was necessary to set up the framework in which 
this work takes place and to state the new ideas it brings to this 
much walked subject. 

We present here a definition of the derived-term automaton of an
expression~$\Ed$ by induction on the 
formation of~$\Ed$, in parallel with the construction of the position 
automaton of~$\Ed$ and with no reference whatsoever to the quotients 
of the series denoted by~$\Ed$ nor to a derivation operation defined 
on expressions.
This new perspective shows that the derived-term automaton is indeed 
intrinsically attached to the structure, or to the syntactic tree, of 
the expression, in the same way as the position automaton is.

The first consequence, or outcome, of the decorrelation between the
construction of the derived-term automaton and the derivation, and
thus the quotient of series, is that it can be achieved on expressions
over \emph{non free} monoids in which the rational languages or
series no longer coincide with the recognisable ones, and hence are
not characterised by the Jacob's theorem quoted above any longer.

The second outcome is that the `morphism theorem' which was somewhat
tedious to establish comes for free with this new point of view as
it is an intrinsic property of the construction: at every step of the
induction, the derived-term automaton is built as a morphic image of
an automaton which is already a morphic image of the position
automaton.

Of course, the connection with the derivation of expressions, and 
the quotient of the denoted series remains when the expression are 
over free monoids, since the new construction yields the same 
automata as the old one.

\medskip

The fact that the derived-term automaton is indeed related to the 
structure of the expression is not completely new.
As we have explained above, and in our own work~\cite{LombSaka05a}, we 
have defined the derived terms by means of an induction rather than 
by the plain derivation.
In~\cite{ChamEtAl09}, and in order to describe an \emph{efficient
algorithm} for the construction of the derived-term automaton, a link
between the {\em positions} of the letters in the rational expression
and the derived terms is made through in-between objects called {\em
c-continuations}.
This allows to define the derived-term automaton as a morphic image of
the position automaton.
This construction is different from ours, since we apply
morphisms at every step of our inductive construction.

There are also been several attempts to apply the derivation 
techniques to expressions over non free monoids, namely direct 
products of free monoids, for dealing with \emph{rational relations}.
In~\cite{Dema17}, the extension is made through a new operator that
represents the direct product of two languages (or relations), while
in~\cite{KonsEtAl21} the atoms of the expression, that are letters in
the classical case, are replaced by pairs of letters or special
symbols (expressing constraints over letters or pairs).
The formalism used in both papers bears some similarities with ours, but
they both use it to define an analogue of the derivation to bring the
construction of the transducer back to the usual construction of the
derived-term automaton.
Notice also that only Boolean transducers are considered
in~\cite{KonsEtAl21}.

\medskip

The essence of the new perspective we take on the derived-term 
automaton of an expression could be described in the classical case of 
the rational expressions on a free monoid.
But we discovered this new point of view when we were dealing with 
\emph{weighted rational expressions on non free monoids}.
Even though it makes the exposition somewhat longer and burdensome, 
we have chosen to present it in its full generality.

In~\secti{pre-not}, we fix the notation for weighted expressions and 
weighted automata and define the morphisms of weighted automoata 
\textit{via} the notion of \emph{conjugacy} which proves to be 
efficient.
In~\secti{sta-aut}, we define the restricted class of \emph{standard 
automata} on which one can lift the rational operators and the 
position automaton, which we prefer to call the \emph{standard 
automaton} of the expression.

The core of the paper lays in~\secti{sta-der-ter} where the new
definition of the derived-term automaton is presented and its
consistency proved.  
Even though the definition goes purely by induction on the formation
of the expression, we have chosen to keep the old terminology which
bears the weight of history and reconnects with it in the prevalent 
case of expressions on a free monoid.
In~\secti{bac-der}, we show, via the notion of \emph{differential} of 
an expression, that the new definition of derived-term automaton 
coincides with the one given in the previous works on the subject, in 
the case of expressions on a free monoid.


\section{{Preliminaries and notation}}%
\label{s.pre-not}%

The definition of usual notions in theoretical computer science, such 
as free monoids, languages, expressions, automata, rational (or 
regular) sets, recognisable sets, etc. may be found in numerous
textbooks.

The corresponding notions of multiplicity (or weight) semirings,
(formal power) series, weighted automata, etc.  are probably less
common knowledge but are still presented in quite a few
books~\cite{BersReut88,BersReut11,DrosEtAl09Ed,SaloSoit77} to which we
refer the reader.
For the notation, we follow~\cite{Saka09,Saka21}.
Let us be more explicit for the two notions we study: the 
\emph{weighted rational expressions} and the \emph{weighted finite 
automata}.
Before, we recall the notions of \emph{graded} monoids and of 
\emph{starrable element} in a semiring.
And then, we define the notion of \emph{morphism} of weighted 
automata that will be instrumental in this work.

The purpose of the paper, is the construction of automata over a
monoid~$M$ which is \emph{not necessarily free}, for instance automata
over~$\Ae\x\Be$ which are transducers --- and this is a key feature of
this work.
At the same time we want the automata possibly be weighted with 
coefficients taken in a semiring~$\K$, thus realising maps from~$M$ 
to~$\K$, that is, \emph{series} in~$\KM$.
The required hypothesis on~$M$ for~$\KM$ to be closed under (Cauchy)
product is that~$M$ be a \emph{graded monoid} (i.e. endowed with an
additive length function) --- which is the case for~$\Ae\x\Be$ for
instance, or more generally for all trace monoids~\cite{DiekRoze95}.

If~$k$ is an element of a semiring~$\K$, $k^{*}$ is the sum of all 
powers of~$k$:
$\msp k^{*}=\sum_{n\in\N}k^{n}\msp$.
This infinite sum may be defined --- $k$~is said to be 
\emph{starrable} --- or not defined --- $k$~is said to be 
\emph{non~starrable}.
We are not interested in the problem of determining whether an element 
of~$\K$ is starrable or not.
Somehow, we consider that~$\K$ is equipped with this 
operator~$\xmd*\xmd$, which is defined on a known subset of~$\K$.
But the question arises to know if we are able, given~$\K$ and~$M$, to 
determine whether a series of~$\KM$ is starrable or not.
The answer is positive, \textit{via} the notion of strong semiring 
(see~\cite{LombSaka05a,Saka09}), and some additional notation.

The \emph{identity element} of~$M$ is denoted by~$\unM$.
We write~$\MPos$ for the set of elements of~$M$ different from~$\unM$, 
that is, the set of elements with a (stricly) positive length:
$\msp\MPos = M\bk\{\unM\}\msp$.

The \emph{constant term}~$\TermCnst{s}$ of a series~$s$ is the 
coefficient of~$\unM$ in~$s$ (that is, the image of~$\unM$ in~$s$).
A series is \emph{proper} if its constant term in~$\zeK$.
The \emph{proper part}~$\spr$ of a series~$s$  is the series obtained 
from~$s$ by zeroing the coefficient of~$\unM$ and keeping all other 
coefficients unchanged.

\begin{definition}
\label{d.str-smr}
A topological semiring is \emph{strong} if the product of two 
summable families is a summable family.
\end{definition}

The definition is taken in view of the following statement.

\begin{samepage}
\begin{theorem}
\label{t.eto-ser-qqu}
Let~$\K$ be a strong semiring and~$M$ a graded monoid.
Let~$s$ be a series of~$\KM$,~$s_{0}=\TermCnst{s}$ its constant term
and~$\spr$ its proper part.
Then~$s^*$ is defined if and only if $s_{0}$~is starrable
and in this case we have
\begin{equation}
s^*  = (s_{0}^*\xmd \spr)^*s_{0}^*
      =  s_{0}^* (\spr\xmd s_{0}^*)^* \eqpnt
\notag
\end{equation}
\end{theorem}
\end{samepage}

The details are not of interest here. 
It is enough for us to know that all usual semirings
such as~$\N$, $\Z$, $\Q$, $\R$, $\C$, $(\Z,\min,+)$, 
\etc are strong topological semirings.
And may be that not all topological semirings are strong (\cf 
\cite{MadoSakaxx}).
In the sequel, the semirings are supposed to be strong, and the 
monoids to be graded, without always stating it explicitely.

\subsection{Weighted rational expressions}

\begin{definition}
\label{d.rat-exp}%
A \emph{rational expression over a monoid~$M$ with weight in a
semiring~$\K$} is a well-formed formula built inductively from the
\emph{constants}~$\zed$ and~$\und$ and the elements~$m$ in~$\MPos$ as
\emph{atomic formulas}, using two binary operators~$\autplus$
and~$\autprod$, one unary operator~$\autstarsymb$ and two operators
for every~$k$ in~$\K$: if~$\Ed$ and~$\Fd$ are expressions, so are
$(k\xmd \Ed)$,
$(\Ed\xmd k)$,
$(\Ed\plusopr\Fd)$,
$(\Ed\prodopr\Fd)$, and
$(\Ed^{*})$.
We denote by~$\KRatEM$ the set of rational expressions over~$M$ with 
weight in~$\K$ and often call them \emph{$\K$-expressions} or even 
simply \emph{expressions}.

Expressions are thus given by the following grammar 
\begin{equation}
\Ed\rightarrow 
\mathsf{0}        \mid
\mathsf{1}        \mid 
m                 \mid
(k\xmd\Ed)\       \mid 
(\Ed\xmd k)\      \mid
(\Ed\plusopr\Ed)  \mid
(\Ed\prodopr\Ed)  \mid 
(\Ed^*)  
\eee
\forall m\in\MPos\quantvrg
\forall k\in\K
\eqpnt
\notag
\end{equation}
\end{definition}

\begin{samepage}
\begin{definition}
\label{d.con-ter-exp}%
The \emph{constant term} of an expression~$\Ed$ in~$\KRatEM$ --- if 
it exists --- is the element of~$\K$, written $\TermCst{\Ed}$, and 
inductively computed using the following equations
\begin{gather}
\TermCst{\zed} = 0 \EqVrgInt
\TermCst{\und} = 1 \EqVrgInt  
\TermCst{m} =0 \quantsmsp \forall m \in \MPos \EqVrgInt 
\notag
\\
\TermCst{k\xmd\Ed} = k\xmd\TermCst{\Ed} \EqVrgInt 
\TermCst{\Ed\xmd k} = \TermCst{\Ed}\xmd k  \quantsmsp \forall k \in \K \EqVrgInt 
\notag
\\
\TermCst{\Fd \plusopr \Gd}= \TermCst{\Fd} + \TermCst{\Gd}\EqVrgInt
\TermCst{\Fd \prodopr \Gd} = \TermCst{\Fd} \xmd \TermCst{\Gd}\EqVrgInt
\notag
\\
\TermCst{\Fd^{*}}= (\TermCst{\Fd})^{*}
\eee \text{\textit{if~$\TermCst{\Fd}$ is starrable}.}
\notag
\end{gather}

If the constant term of a subexpression~$\Fd$ of~$\Ed$ is not
starrable, $\TermCst{\Ed}$ is \emph{undefined}, and~$\Ed$ is said to
be \emph{non valid}; otherwise,~$\Ed$ is a \emph{valid} expression.
\end{definition}
\end{samepage}

\begin{samepage}
\begin{definition}
\label{d.den-ser}%
With every \emph{valid} expression~$\Ed$ in~$\KRatEM$ is associated a 
series of~$\KM$, which is called \emph{the series denoted by~$\Ed$},
and which we write~$\CompExpr{\Ed}\msp$.\\
\ee The series~$\CompExpr{\Ed}$ is inductively defined by
\begin{gather}
\CompExpr{\zed}=\zeK \EqVrgInt
\CompExpr{\und}= \unM \EqVrgInt  
\CompExpr{m}= m\quantsmsp \forall m \in \MPos \EqVrgInt 
\CompExpr{k\xmd\Ed} = k~\CompExpr{\Ed} \EqVrgInt 
\CompExpr{\Ed\xmd k} = \CompExpr{\Ed}~ k  \quantsmsp \forall k \in \K \EqVrgInt 
\notag
\\
\CompExpr{\Fd \plusopr \Gd}= \CompExpr{\Fd} + \CompExpr{\Gd}\EqVrgInt
\CompExpr{\Fd \prodopr \Gd} = \CompExpr{\Fd} ~ \CompExpr{\Gd}\EqVrgInt
\e\text{and}
\notag
\\
\CompExpr{\Fd^{*}}= (\CompExpr{\Fd})^{*}
\ee \text{($\CompExpr{\Fd}$ is starrable by the validity of~$\Ed$ 
and \theor{eto-ser-qqu}).}
\notag
\end{gather}
Two expressions are \emph{equivalent} if they denote the same series.
\end{definition}
\end{samepage}

It directly follows from Definitions~\ref{d.con-ter-exp} 
and~\ref{d.den-ser} that the constant term of an expression is 
equal to the constant term of the series denoted by the expression.

\begin{proposition}
\label{p.cst-trm}
\ee
$\msp \TermCnst{\Ed} = \TermCnst{\CompExpr{\Ed}}\msp$.
\end{proposition}

\setcounter{exa-exp}{\value{theorem}}
\begin{example}
\label{e.exa-exp}
The $\Z$-expression~$\Ed_{\idex}$ over the monoid~$\{a,b\}^{*}$, 
$\msp\Ed_{\idex}= a^*\matmul(a^* + (-1)b^*)^*\msp$, is valid:
$\msp\TermCnst{\Ed_{\idex}}=1\msp$.
\end{example}

Even though they do not play a role in this work, the 
definition of the set of \emph{rational series} and its 
characterisation with expressions build its background.

\begin{definition}
\label{d.rat-ser}
The set of \emph{$\K$-rational series over~$M$} is the smallest
subalgebra of~$\KM$ which contains the polynomials and is closed under
star.
It is denoted by~$\KRatM$.
\end{definition}

\begin{proposition}
\label{p.rat-ser}
A series of~$\KM$ is rational if and only if it is denoted by a valid 
expression in~$\KRatEM$.
\end{proposition}

\subsection{Weighted finite automata}

An automaton~$\Ac$ over a monoid~$M$ with weights in a semiring~$\K$ 
is a \emph{labelled directed graph}~$(Q,E)$, together with two 
functions~$I$ and~$T$ from the set~$Q$ of vertices --- called 
\emph{states} --- into~$\K$.
The set~$E$ of edges --- called \emph{transitions} --- is contained 
in~$Q\x\K\x\MPos\x Q$, that is, every transition is labelled with a 
monomial~$k\xmd m$ --- the \emph{weighted label} of the transition 
--- where~$k$ is the \emph{weight} of the transition and~$m$ its 
\emph{label}.
The automaton~$\Ac$ is \emph{finite} if~$E$ is finite.

The weighted label of a path in~$\Ac$ is the 
product of the weighted labels of the transitions that form the path, 
hence a monomial~$h\xmd x$, where~$h$ is the product of the weights 
of the transitions and~$x$ the product of their labels.

The automaton~$\Ac$ determines a map from~$M$ to~$\K$, that is a 
series in~$\KM$, called the \emph{behaviour} of~$\Ac$ and denoted 
by~$\CompAuto{\Ac}\msp$.
The series~$\CompAuto{\Ac}$ maps every~$x$ in~$M$ to the \emph{sum} of
all elements~$\msp I(p)\xmd h\xmd T(q)\msp$ where~$h$ is the weight 
of a path~$\pi$ with label~$x$, for all such paths~$\pi$ from~$p$ to~$q$, 
and all pairs of states~$(p,q)$.
The definition of~$\CompAuto{\Ac}$ takes a handier form in an algebraic 
setting.
The set~$E$ of transitions of~$\Ac$ is conveniently described by the 
\emph{transition matrix} of~$\Ac$, which we also denote by~$E$ (as it 
will be indeed the unique way we deal with this set in the sequel), 
and which is thus a matrix of dimension~$Q\x Q$ whose $(p,q)$-entry 
is the sum of the weighted labels of the transitions that go from~$p$ 
to~$q$, a linear combination of elements of~$\MPos$ when~$\Ac$ is 
finite.
We write $\msp\Ac=\autiet\msp$ where the function~$I$ is written as a
row-vector of dimension~$Q$ whose $p$th entry is~$I\!(p)$ and the
function~$T$ is written as a column-vector of dimension~$Q$ whose
$q$th entry is~$T\!(q)$.
Since the formation of paths corresponds to the multiplication of the 
transition matrix, the behaviour of~$\Ac$ may then be written as
\begin{equation}
\CompAuto{\Ac} = I\matmul E^{*}\matmul T
\eqpnt
\notag
\end{equation}

Two automata are \emph{equivalent} if they have the same behaviour.
Finite automata and rational expressions have the same computational 
power, as expressed by the following statement.

\begin{theorem}
\label{t.fta-wgt-mon}%
Let~$M$ be a graded monoid.
A series of~\/$\KM$ is rational if and only if it is 
the behaviour of a finite $\K$-automaton over~$M$.
\end{theorem}

The subject of this work is the study of a particular proof of the 
sufficient condition of this statement.

\subsection{Morphisms and quotient of weighted automata}

Automata are structures; one can thus define \emph{morphisms} between 
them.
We choose to define the morphisms of weighted automata
via the notion of \emph{conjugacy},
borrowed from the theory of symbolic dynamical systems.
It is the most concise way, and ideally suited for the sequel.

\begin{definition}
A $\K$-automaton 
$\msp\Ac=\autiet\msp$ \emph{is conjugate to}  
a $\K$-automaton 
$\msp\Bc=\autjfu\msp$
if there exists a matrix $X$ with entries in $\K$ such that
\begin{equation}
I\xmd X=J, \ee 
E\xmd X=X\xmd F, \e \text{and} \e
T=X\xmd U.
\label{q.con-aut}
\end{equation}
The matrix $X$ is the \emph{transfer matrix} of the conjugacy
and we write $\msp\Ac\ConjAuto{X}\Bc\msp$.
\end{definition}

If~$\Ac$ is conjugate to~$\Bc$, 
then, for every~$n$, the series of equalities holds:  
\begin{equation*}
I\mmul E^{n}\mmul T =
I\mmul E^{n}\mmul X\mmul U =
I\mmul E^{n-1}\mmul X\mmul F \mmul U = \ldots =
I\mmul X\mmul F^{n} \mmul U = J \mmul F^{n} \mmul U
\eqvrg 
\end{equation*}
from which 
$\msp I\mmul E^{*}\mmul T = J \mmul F^{*} \mmul U\msp$
directly follows.

\begin{proposition}
\label{p.cnj-equ}%
If
$\msp\Ac\msp$ {is conjugate to} 
$\msp\Bc\msp$, then $\msp\Ac\msp$ and  $\msp\Bc\msp$
are equivalent.
\end{proposition}

Let $\msp\varphi\colon Q\rightarrow R \msp$ be a \emph{surjective} map 
and~$\xphi$ the $Q\x R$-matrix where 
the $(q,r)$-th entry is~$1$ if $\varphi(q)=r$, and~$0$ otherwise.
Since~$\varphi$ is a map, every row of~$\xphi$ contains exactly 
one~$1$ and
since~$\varphi$ is surjective, every column of~$\xphi$ contains
at least one~$1$.
Such a matrix is called an \emph{amalgamation matrix} in the setting 
of symbolic dynamics.
By convention, if we deal with $\K$-automata, an amalgamation matrix 
is silently assumed to be a $\K$-matrix, that is, the null entries 
are equal to~$\zeK$ and the non zero entries to~$\unK$.

\begin{definition}
\label{d.out-mor}%
Let~$\Ac$ and~$\Bc$\/ be 
two $\K$-automata of dimension~$Q$ and~$R$ respectively.
We say that a surjective map 
$\msp\varphi\colon Q\rightarrow R \msp$
is 
a \emph{morphism}\footnote{%
   The morphisms of weighted automata are `more constrained' than 
   those of Boolean automata.
   They correspond to what is often called \emph{simulation}.
   Moreover, they are \emph{directed} and for a same map~$\varphi$ 
   one should distinguish between \emph{Out}-morphism and 
   \emph{In}-morphism.
   But in this work we only deal with Out-morphisms, which we simply 
   call morphisms.} 
(from~$\Ac$ onto~$\Bc$)
if~$\Ac$ is conjugate to~$\Bc$ by~$\xphi$, \ie
if $ \msp\Ac\ConjAuto{\xphi}\Bc\msp$, and we write
$\msp\varphi\colon\Ac\rightarrow\Bc\msp$.

We also say that~$\Bc$ is a \emph{quotient} of~$\Ac$, if there exists
a morphism~$\msp\varphi\colon\Ac\rightarrow\Bc\msp$.
\end{definition}

The composition of two morphisms is a morphism.
From \propo{cnj-equ} follows that any quotient of~$\Ac$ is equivalent 
to~$\Ac$.

On the other hand, we can determine whether a surjective map 
$\msp\varphi\colon Q\rightarrow R \msp$
is a morphism or not, without reference to any automaton~$\Bc$.
From~$\xphi$ we construct a \emph{selection matrix}~$\yphi$ by 
transposing~$\xphi$ and by zeroing some of its non zero entries in 
such a way that~$\yphi$ is row-monomial, with \emph{exactly one}~$1$ 
per row.
A matrix~$\yphi$ is not uniquely determined by~$\varphi$ but also 
depends on the choice of a `representative' in each class of the map 
equivalence of~$\varphi$.

\begin{proposition}
\label{p.out-mor}%
Let~$\msp\Ac=\autiet\msp$ be a $\K$-automaton of dimension~$Q$.
Let 
$\varphi\colon Q\rightarrow R $
be a surjective map, $\xphi$ its amalgamation matrix, and~$\yphi$ a 
selection matrix.
Then~$\varphi$ is a \emph{morphism} if~$\Ac$ is conjugate 
by~$\xphi$ to the automaton~$\varphi(\Ac)$ of dimension~$R$:
$\varphi(\Ac) = 
\aut{I\matmul\xphi, \msp\yphi\matmul E\matmul\xphi, \msp\yphi\matmul T}$ 
(in which case~$\varphi(\Ac)$ does not depend on the 
choice of~$\yphi$).
\end{proposition}

This proposition points out that the image of~$\varphi$ is indeed
immaterial and what only counts, and makes it a morphism of automata
or not, is the map equivalence of~$\varphi$.

In the proofs at \secti{sta-der-ter}, we use indeed a more intuitive
description of morphisms.
With the same notation as above, the map~$\varphi$ is a morphism if
and only if the rows of~$E\matmul\xphi$ whose indices are equivalent
modulo~$\varphi$ are equal and the entries of~$T$ whose indices are
equivalent modulo~$\varphi$ are equal.


\section{{Standard automata}}%
\label{s.sta-aut}%

We define a restricted class of automata, and then show that rational
operations on series can be lifted on the automata of that class.
They are thus well-suited for the constructions we build by induction
on the formation of the expressions.

An automaton is \emph{standard} if it has only one initial state,
which is the end of no transition.
\figur{sta-aut} shows a standard automaton, both as a sketch, and 
under the matrix form.
The definition does not forbid the initial state~$i$ from also being
final and the scalar~$c$ in~$\K$, is the \emph{constant term} of
$\msp\CompAuto{\Ac}\msp$.

\begin{figure}[ht]
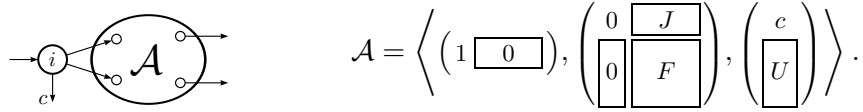

\centering
\TinyPicture%
\VCDraw{%
\begin{VCPicture}{(-4.5,-1.3)(2.5,1.3)}%
\VCPut{(0,0)}%
  {%
   \BigAuto{\Ac}%
   \State[i]{(-3,0)}{A}%
   \VSState{(-1,0.646)}{B1}\VSState{(-1,-0.646)}{B2}%
   \VSState{(1,.736)}{C2}\VSState{(1,-0.736)}{C3}%
  }%
\Initial[w]{A}\FinalR{s}{A}{c}%
\Final[e]{C2}\Final[e]{C3}%
\Edge{A}{B1}\Edge{A}{B2}%
\end{VCPicture}%
}%
\SmallPicture
\eee
$\msp
\begin{displaystyle}
\Ac =
\StanBloc{J}{F}{c}{U}
\end{displaystyle}
\msp$.
\caption{A standard automaton}
\label{f.sta-aut}
\end{figure}

\noindent
Elementary matrix computations show
\begin{equation}
    \CompAuto{\Ac} = c + J\mmul F^{*}\mmul U
	\eqvrg
\label{q.sta-aut-beh}
\end{equation}
where
$\msp c=\TermCnst{\CompAuto{\Ac}}\msp$
is the \emph{constant term} of~$\CompAuto{\Ac}$ and
$\msp J\mmul F^{*}\mmul U\msp$
is the proper part of $\CompAuto{\Ac}\xmd$.

It is convenient to say --- when there is no ambiguity --- that the 
\emph{dimension} of the standard automaton~$\Ac$ is is the dimension 
of the vector~$J$ (or of the matrix~$F$).

It is rather obvious that every automaton is
equivalent to a standard one, but this will not be used here.

\subsection{Operations on standard automata}
\label{s.ope-sta-aut}%

Their special form allows to define \emph{operations} on standard
automata that are parallel to the \emph{rational operations}.
Let~$\Ac$ (as in~\figur{sta-aut}) and~$\Bc$ (with obvious 
notation) be two standard ($\K$-)automata; let~$k$ be in~$\K$.
Then we define the following standard automata
{\allowdisplaybreaks
\begin{alignat}{2}
\pointn&\e &
k\xmd{\Ac} &=
\StanBloc{k\xmd J}{F}{k\xmd c}{U}\eqvrg
\label{q.sta-aut-lml}
\\[1ex]
\pointn& &
{\Ac}\xmd k &=
\StanBloc{J}{F}{c\xmd k}{U\xmd k}\eqvrg
\label{q.sta-aut-rml}
\\[1ex]
\pointn& &
    {\Ac}+{\Bc} &=
\aut{\redmatu{\lignetblblvs{1}{0}{0}},
     \redmatu{\matricettblblvs{0}{J}{K}%
                              {0}{F}{0}%
                              {0}{0}{G}},
     \redmatu{\vecteurtblblvs{c+d}{U}{V}}}\eqvrg
\eee\ee
\label{q.sta-aut-sum}
\\[1ex]
\pointn& &
    {\Ac}\matmul{\Bc} &= 
\aut{\redmatu{\lignetblblvs{1}{0}{0}},
     \redmatu{\matricettblblvs{0}{J}{c\xmd K}%
                              {0}{F}{U\matmul K}%
                              {0}{0}{G}},
     \redmatu{\vecteurtblblvs{c\xmd d}{U\xmd d}{V}}}\eqvrg
\label{q.sta-aut-pro}
\\
\intertext{%
and finally ${\Ac}^{*}$, which is defined when~$c^{*}$ is defined, 
}
\pointn& &
{\Ac}^{*} &=
\StanBloc{c^{*}\xmd J}{H}{c^{*}}{U\xmd c^{*}}
\e \text{with $\msp H= U\matmul c^{*}\xmd J + F\msp$.}
\label{q.sta-aut-sta}
\end{alignat}
}

After these definitions, it is rather natural to say that the two 
exterior multiplications and the star are the \emph{dimension 
invariant} operations.
Elementary matrix computations then establish the following.

\begin{proposition}
\label{p.sta-aut-cmp}
Let~$\Ac$ and~$\Bc$ be two standard $\K$-automata, and let~$k$ be an 
element in~$\K$.
It then holds:
$\msp\CompAuto{k\xmd\Ac} = k\xmd\CompAuto{\Ac}\msp$,
$\msp\CompAuto{\Ac\xmd k} =\CompAuto{\Ac}\xmd k\msp$,
$\msp\CompAuto{\Ac + \Bc} =\CompAuto{\Ac} + \CompAuto{\Bc}\msp$,
and $\msp\CompAuto{\Ac \matmul \Bc} 
=\CompAuto{\Ac}~\CompAuto{\Bc}\msp$.
\end{proposition}

The case of the star operation is significantly more involved and 
requires that \theor{eto-ser-qqu} is first established.
Then, the following statement holds.

\begin{proposition}
\label{p.sta-aut-str}
Let~$\K$ be a strong semiring.
If~$\Ac$ is a standard $\K$-automaton, it then holds:
$\msp\CompAuto{{\Ac}^{*}} = \left(\CompAuto{\Ac}\right)^{*}\msp$.
\end{proposition}

\subsection{The standard automaton of an expression}
\label{s.sta-aut-exp}

The definition of the `rational' operations on standard automata 
immediately induces the definition of a standard automaton 
canonically associated with every rational expression.
It coincides with the automaton first defined by Glushkov 
in~\cite{Glus61}.

\begin{proposition}
\label{p.sta-aut-exp}
For every valid rational $\K$-expression~$\Ed$, there exists a 
canonical standard $\K$-automaton~$\Stan{\Ed}$ that realises the 
series denoted by~$\Ed$, that is, 
$\msp \CompAuto{\Stan{\Ed}} = \CompAuto{\Ed}\msp$.
\end{proposition}

\begin{proof}
The definition of~$\Stan{\Ed}$ starts with the definitions of 
standard automata for the atomic formulas.

\pointn
$\msp \Ed = \zed \msp$
then
$\msp 
\begin{displaystyle}
\Stan{\zed} = \aut{\redmatu{\matriceuu{1}},
                   \redmatu{\matriceuu{0}},
                   \redmatu{\matriceuu{0}}} 
\end{displaystyle}
\msp$.

\pointn
$\msp \Ed = \und \msp$
then
$\msp 
\begin{displaystyle}
\Stan{\und} = \aut{\redmatu{\matriceuu{1}},
                   \redmatu{\matriceuu{0}},
                   \redmatu{\matriceuu{1}}} 
\end{displaystyle}
\msp$.

\pointn
$\msp \Ed = m \in M \msp$
then
$\msp 
\begin{displaystyle}
\Stan{m} = 
\aut{\redmatu{\ligned{1}{0}},
          \redmatu{\matricedd{0}{m}{0}{0}},
          \redmatu{\vecteurd{0}{1}}} 
\end{displaystyle}
\msp$.\\[1ex]

\noindent
It is clear that 
$\msp \CompAuto{\Stan{\zed}} = 0 = \CompAuto{0}\msp$,
$\msp \CompAuto{\Stan{\und}} = 1 = \CompAuto{\und}\msp$, and
$\msp \CompAuto{\Stan{m}} = m = \CompAuto{m}\msp$.

The natural definitions:
$\msp \Stan{k\xmd\Ed} = k\xmd\Stan{\Ed}\msp$,
$\msp \Stan{\Ed\xmd k} = \Stan{\Ed}\xmd k\msp$,
$\msp \Stan{\Fd +\Gd} = \Stan{\Fd} + \Stan{\Gd}\msp$,
$\msp \Stan{\Fd \matmul\Gd} = \Stan{\Fd}\matmul\Stan{\Gd}\msp$, and
$\msp \Stan{\Fd^{*}} = \left(\Stan{\Fd}\right)^{*}\msp$,
allow the construction of~$\Stan{\Ed}$ for every valid $\K$-rational 
expression~$\Ed$, by induction on the formation of the expression, 
whereas Propositions~\ref{p.sta-aut-cmp} and~\ref{p.sta-aut-str} insure that 
$\msp\CompAuto{\Stan{\Ed}}=\CompAuto{\Ed}\msp$.
\end{proof}

This automaton~$\StanE$ will be indifferently called `\emph{the} 
standard automaton' or `the position automaton' of~$\Ed$.

\setcounter{theor-tmp}{\value{theorem}}
\setcounter{theorem}{\value{exa-exp}}
\begin{example}[continued]
Let us write~$\Ed_{\idex}=\Fd_{\idex}\matmul\Gd_{\idex}$ with
$\Fd_{\idex}=a^*$ and $\Gd_{\idex}=(a^* + (-1)b^*)^*$.
It comes
\begin{gather}
\Stan{\Fd_{\idex}}=
   \aut{\ligned{1}{0},\matricedd{0}{a}{0}{a},\vecteurd{1}{1}}\EqVrgInt
\e
\Stan{\Gd_{\idex}}=
\aut{\lignet{1}{0}{0},
     \matricett{0}{a}{-b}{0}{2a}{-b}{0}{a}{0},
	 \vecteurt{1}{1}{1}}\EqVrgInt
\notag
\\[2ex]
\text{and}\ee
\Stan{\Ed_{\idex}}=\Stan{\Fd_{\idex}}\matmul\Stan{\Gd_{\idex}}=
\aut{%
\begin{pmatrix} 1 & 0 & 0 & 0\end{pmatrix},%
\begin{pmatrix}%
 0 & a & a & -b\\%
 0 & a & a & -b\\%
 0 & 0 & 2a & -b\\%
 0 & 0 & a & 0%
 \end{pmatrix},%
\begin{pmatrix} 1\\ 1 \\ 1 \\ 1\end{pmatrix}}
\eqpnt
\notag
\end{gather}
\end{example}
\setcounter{theorem}{\value{theor-tmp}}

\subsection{Morphisms of standard automata}
\label{s.mor-sta-aut}%

Let~$\Ac$ and~$\Ac'$ be two standard automata,
\begin{equation}
\Ac =
\StanBloc{J}{F}{c}{U}
\EqVrgInt
\text{and}
\msp
\Ac' =
\StanBloc{J'}{F'}{c'}{U'}
\!\eqvrg
\notag
\end{equation}
and
$\msp \varphi\colon\Ac\rightarrow\Ac' \msp$
a morphism of automata.
The image by~$\varphi$ of the inital state~$i$ of~$\Ac$ is 
necessarily the initial state~$i'$ of~$\Ac'$ and no other state~$q$ 
of~$\Ac$ is mapped onto~$i'$ for otherwise~$q$ would not be 
accessible. 

Hence the transfer matrix of~$\varphi$ is of the form
\begin{equation}
\redmatu{\matriceddblvs{1}{0}{0}{X_{\varphi}}}
\eqvrg
\notag
\end{equation}
and the conjugacy relation \equnm{con-aut} passes to the `core' of the
automata
\begin{equation}
c = c' \EqVrgInt
J\xmd X_{\varphi}=J'\EqVrgInt
F\xmd X_{\varphi}=X_{\varphi}\xmd F'\EqVrgInt \text{and} \e
U=X_{\varphi}\xmd U'
\eqpnt 
\notag
\end{equation}

The operations on standard automata that we have defined above are 
consistant with morphisms.
 
\begin{proposition}
\label{p.mor-sta}
Let~$\Ac$ and~$\Ac'$, $\Bc$ and~$\Bc'$ be four standard automata, and 
let 
$\msp \varphi\colon\Ac\rightarrow\Ac' \msp$ and
$\msp \psi\colon\Bc\rightarrow\Bc' \msp$
be two automata morphisms.
Then
$\msp \varphi\x\psi\colon\Ac+\Bc\rightarrow\Ac'+\Bc'\msp$,
$\msp\varphi\x\psi\colon\Ac\matmul\Bc\rightarrow\Ac'\matmul\Bc'\msp$, 
and $\msp \varphi\colon\Ac^{*}\rightarrow(\Ac')^{*} \msp$
are automata morphisms.
\end{proposition}

\begin{proof}
The statement for the addition is obvious. 
The computations for the other two operations are hardly more complex.
\setlength{\longueurblc}{8.5ex}
\setlength{\hauteurblc}{6.5ex}
\setlength{\centrageblc}{-2.5ex}
\renewcommand{\redmatu}[1]{\scalebox{0.76}{#1}}

For the product we have
\begin{align}
\redmatu{\matricettblblvs{0}{J}{c\mmul K}%
                {0}{F}{U\mmul K}%
                {0}{0}{G}}
\matmul
\redmatu{\matricettblblvs{1}{0}{0}%
                {0}{X_{\varphi}}{0}%
                {0}{0}{X_{\psi}}}
& =
\redmatu{\matricettblblvs{0}{J\mmul X_{\varphi}}{c\mmul K\mmul X_{\psi}}%
                {0}{F\mmul X_{\varphi}}{U\mmul K\mmul X_{\psi}}%
                {0}{0}{G\mmul X_{\psi}}}
=
\notag
\\[2ex]
\redmatu{\matricettblblvs{0}{J'}{c\mmul K'}%
                {0}{X_{\varphi}\mmul F}{X_{\varphi}\mmul U'\mmul K'}%
                {0}{0}{X_{\psi}\mmul G'}}
& =
\redmatu{\matricettblblvs{1}{0}{0}%
                {0}{X_{\varphi}}{0}%
                {0}{0}{X_{\psi}}}
\matmul
\redmatu{\matricettblblvs{0}{J'}{c\mmul K'}%
                {0}{F'}{U'\mmul K'}%
                {0}{0}{G'}}
\EqPnt
\notag
\end{align}

\setlength{\longueurblc}{6.5ex}%
\setlength{\hauteurblc}{5ex}%
\setlength{\centrageblc}{-2ex}%
\renewcommand{\redmatu}[1]{\scalebox{0.84}{#1}}%

\noindent
And for the star, the sequence of equalities
\begin{equation}
(U\mmul c^{*} J + F)\mmul X_{\varphi} =
(U\mmul c^{*} J)\mmul X_{\varphi} + F \mmul X_{\varphi} =
X_{\varphi}\mmul(U'\mmul c^{*} J') + X_{\varphi}\mmul F =
X_{\varphi}\mmul(U'\mmul c^{*} J' + F') 
\notag
\end{equation}
yields the result.
\end{proof}


\section{{The standard derived-term automaton of an expression}}%
\label{s.sta-der-ter}

By a process similar to the construction of~$\Stan{\Ed}$, though more
involved, we associate now with every $\K$-expression~$\Ed$
\emph{another} standard automaton, the \emph{standard derived-term
automaton}~$\TermE$.
We begin with the definition of the set~$\DerTer{\Ed}$ of 
\emph{derived terms} of~$\Ed$.

\subsection{The derived terms of an expression}
\label{s.der-ter-exp}

The set of derived terms is defined by induction on the formation of
the expression.

\begin{definition}
\label{d.der-ter}
The set~$\DerTerE$ of derived terms\footnote{%
   The definition is the same as in~\cite{LombSaka05a} and all 
   subsequent works of ours.
   We have changed the name from \emph{true derived term} to 
   \emph{derived term} and the
   notation from~$\TruDerTer{\Ed}$ to~$\DerTerE$ for 
   simplification, as the new presentation allows it.}
of a $\K$-expression~$\Ed$ over~$M$ is a set of $\K$-expressions 
defined inductively by:\\
{\allowdisplaybreaks
\noindent
\textbf{Base cases}
\begin{alignat}{3}
\pointn&\ee&
\Ed &= \zed \e\text{or}\e \Ed = \und &\ee
\DerTerE &=\emptyset\eqpnt
\eee\eee\eee\ee\e
\label{q.dtr-0-1}
\\
\pointn&\ee&
\Ed &= m \e m\in M\bk\unM\e  &
\DerTerE &=\{\und\}\eqpnt
\label{q.dtr-m}
\end{alignat}
\textbf{Induction}
\begin{alignat}{3}
\pointn&\ee&
\Ed &= k\xmd \Fd &
\DerTerE &= \DerTerF\eqpnt
\label{q.dtr-k-lft}
\\
\pointn&\ee&
\Ed &= \Fd\xmd k &
\DerTerE &= \DerTerF\xmd k = \Defi{\Kd\xmd k}{\Kd\in\DerTerF} \eqpnt
\label{q.dtr-k-rgt}
\\
\pointn&\ee&
\Ed &= \Fd + \Gd &\ee
\DerTerE &= \DerTerF \cup \DerTerG\eqpnt
\label{q.dtr-sum}
\\
\pointn&\ee&
\Ed &= \Fd \matmul \Gd &
\DerTerE &= \DerTerF\matmul\Gd \cup \DerTerG 
          = \Defi{\Kd\matmul\Gd}{\Kd\in\DerTerF} \cup \DerTerG \eqpnt
\ee
\label{q.dtr-prd}
\\
\pointn&\ee&
\Ed &= \Fd^{*} &
\DerTerE &= \DerTerF\matmul\Fd^{*} 
          = \Defi{\Kd\matmul\Fd^{*}}{\Kd\in\DerTerF} \eqpnt
\label{q.dtr-str}
\end{alignat}
}
\end{definition}

\begin{lemma}
\label{l.dtr-car}
Let~$\Ed$ be a $\K$-expression over~$M$.
Then $\msp\jsCard{\DerTerE}\leq\LittLeng{\Ed}\msp$.
\end{lemma}

\begin{proof}
The equality holds for the base cases.
Both litteral length and number of derived terms are invariant for 
dimension invariant operations.
For the addition and product operations, 
$\msp\LittLeng{\Fd+\Gd}=\LittLeng{\Fd\matmul\Gd}
=\LittLeng{\Fd}+\LittLeng{\Gd}\msp$ 
and $\DerTer{\Fd+\Gd}$ and $\DerTer{\Fd\matmul\Gd}$ are the union of 
sets each of which satisfies the inequality: they also satisfy the 
inequality, \emph{all the more that the union may not be disjoint}.
\end{proof}

Indeed, the interest, the subtility, and the difficulty, of the 
construction to come arise from the fact that the union in the 
definition of $\DerTer{\Fd+\Gd}$ and $\DerTer{\Fd\matmul\Gd}$ happens 
not to be disjoint.

\setcounter{theor-tmp}{\value{theorem}}
\setcounter{theorem}{\value{exa-exp}}
\begin{example}[Continued]
Let 
$\msp\Fd_{\idex}=a^*\msp$,
$\msp\Gd_{\idex}=(a^* + (-1)b^*)^*\msp$ and
$\msp\Ed_{\idex}=\Fd_{\idex}\matmul\Gd_{\idex}\msp$.
It holds: 
$\msp\DerTer{\Fd_{\idex}}=\{a^*\}\msp$,
$\msp\DerTer{\Gd_{\idex}}=\{a^*\matmul\Gd_{\idex},b^*\matmul\Gd_{\idex}\}\msp$
and $\msp\DerTer{\Ed_{\idex}}=\DerTer{\Gd_{\idex}}\msp$.
\end{example}
\setcounter{theorem}{\value{theor-tmp}}

\subsection{The inductive definition of the standard derived-term automaton}
\label{s.ind-con-sta}

With every $\K$-expression~$\Ed$, and by induction on the formation 
of~$\Ed$, we associate a standard automaton~$\TermE$ of 
dimension~$\DerTerE$, which we call the \emph{standard derived-term 
automaton} of~$\Ed$.

\noindent
\textbf{Base cases}

$\msp\Term{\zed}= \Stan{\zed}\msp$,
$\msp\Term{\und}= \Stan{\und}\msp$,
and $\msp\Term{m}= \Stan{m}\msp$.
for every~$m$ in~$M$.

\noindent
\textbf{Dimension invariant operations}

$\msp\Term{k\xmd\Fd}= k\xmd\TermF\msp$,
$\msp\Term{\Fd\xmd k}= \TermF\xmd k\msp$,
and $\msp\Term{\Fd^{*}}= \left(\TermF\right)^{*}\msp$.

\noindent
\textbf{Addition and product}

\pointn
$\msp\TermF+\TermG\msp$
is 
a standard automaton of dimension
$\msp\DerTerF\sqcup\DerTerG\msp$.
Let~$\varphi$ be the `natural' map
\begin{equation}
\varphi \colon \DerTerF\sqcup\DerTerG \rightarrow
              \DerTerF\cup\DerTerG
\eqvrg
\notag
\end{equation}
that is, $\varphi$ maps two terms~$\Kd$ and~$\Kd'$ 
of~$\DerTerF\sqcup\DerTerG$ onto one if they are \emph{equal}, hence 
$\msp\Kd\in\DerTerF\msp$,
$\msp\Kd'\in\DerTerG\msp$ and 
$\msp\Kd=\Kd'\msp$, or, to state it otherwise, if
$\msp\Kd\in\DerTerF\cap\DerTerG\msp$.

\begin{proposition}
\label{p.add-trm}%
The map~$\varphi$ is a morphism of automata.
\end{proposition}

\noindent
And we define
\begin{equation}
  \Term{\Fd+\Gd}=\varphi(\TermF+\TermG)
\eqpnt
\notag
\end{equation}

\pointn
$\msp\TermF\matmul\TermG\msp$
is 
a standard automaton of dimension
$\msp\DerTerF\sqcup\DerTerG\msp$ in bijection with
$\msp\DerTerF\matmul\Gd\sqcup\DerTerG\msp$.
Let~$\psi$ be the `natural' map
\begin{equation}
\psi \colon \DerTerF\matmul\Gd\sqcup\DerTerG \rightarrow
              \DerTerF\matmul\Gd\cup\DerTerG
\eqvrg
\notag
\end{equation}
that is, $\varphi$ maps two terms~$\Kd$ and~$\Kd'$ 
of~$\DerTerF\matmul\Gd\sqcup\DerTerG$ onto one if they are \emph{equal}, hence 
$\msp\Kd\in\DerTerF\matmul\Gd\msp$,
$\msp\Kd'\in\DerTerG\msp$ and 
$\msp\Kd=\Kd'\msp$, or, to state it otherwise,
$\msp\Kd\in\DerTerF\matmul\Gd\cap\DerTerG\msp$.

\begin{proposition}
\label{p.pro-trm}%
The map~$\psi$ is a morphism of automata.
\end{proposition}

\noindent
And we define
\begin{equation}
  \Term{\Fd\matmul\Gd}=\psi(\TermF\matmul\TermG)
\eqpnt
\notag
\end{equation}

This ends the inductive definition of~$\TermE$.
Modulo the proof of Propositions~\ref{p.add-trm} and~\ref{p.pro-trm}
which is given below, this definition directly implies the following
key statement of the paper, by induction on the formation of the
expression~$\Ed$ and as a consequence of 
Propositions~\ref{p.sta-aut-exp}, \ref{p.mor-sta},
and~\ref{p.cnj-equ}.

\begin{theorem}
\label{t.trm-aut-exp}%
For every valid $\K$-rational expression~$\Ed$, the
standard $\K$-automaton~$\TermE$ realises the 
series denoted by~$\Ed$ and is a quotient of~$\StanE$.
\EoP
\end{theorem}

\setcounter{theor-tmp}{\value{theorem}}
\setcounter{theorem}{\value{exa-exp}}
\begin{example}[continued]
Let 
$\msp\Fd_{\idex}=a^*\msp$,
$\msp\Gd_{\idex}=(a^* + (-1)b^*)^*\msp$ and
$\msp\Ed_{\idex}=\Fd_{\idex}\matmul\Gd_{\idex}\msp$.
We have seen that 
$\msp\DerTer{\Fd_{\idex}}=\{a^*\}\msp$,
$\msp\DerTer{\Gd_{\idex}}=\{a^*\matmul\Gd_{\idex},b^*\matmul\Gd_{\idex}\}\msp$
and $\msp\DerTer{\Ed_{\idex}}=\DerTer{\Gd_{\idex}}\msp$.

It then comes:
$\msp\Term{\Fd_{\idex}}=\Stan{\Fd_{\idex}}\msp$,
$\msp\Term{\Gd_{\idex}}=\Stan{\Gd_{\idex}}\msp$ and
\begin{equation}
\Term{\Fd_{\idex}}\matmul\Term{\Gd_{\idex}}=
\aut{%
\begin{pmatrix} 1 & 0 & 0 & 0\end{pmatrix},%
\begin{pmatrix}%
 0 & a & a & -b\\%
 0 & a & a & -b\\%
 0 & 0 & 2a & -b\\%
 0 & 0 & a & 0%
 \end{pmatrix},%
\begin{pmatrix} 1\\ 1 \\ 1 \\ 1\end{pmatrix}}
\eqpnt
\notag
\end{equation}

The derived term of~$\Fd_{\idex}$, $\msp a^{*}\msp$, multiplied 
by~$\Gd_{\idex}$, is equal to the first derived term 
of~$\Gd_{\idex}$, $\msp a^{*}\matmul\Gd_{\idex}\msp$.
They index respectively the second and third rows and columns of the 
matrix above.
If we add the second and third columns, we get a matrix whose second 
and third rows are equal, and the second and third entries of the 
final vector are also equal (instance of \propo{pro-trm}).
These two states may then be merged to build the quotient and we get
\begin{equation}
\Term{\Ed_{\idex}}=
\aut{\lignet{1}{0}{0},
     \matricett{0}{2a}{-b}{0}{2a}{-b}{0}{a}{0},
	 \vecteurt{1}{1}{1}}
\eqpnt
\notag
\end{equation}
\end{example}
\setcounter{theorem}{\value{theor-tmp}}

\subsection{Proof of Propositions~\ref{p.add-trm} and~\ref{p.pro-trm}}
\label{s.pro-ofs}

The construction of~$\TermE$ starts with the same automata 
as~$\StanE$ for the base cases.
At every step, it uses an operation on standard automata, and 
possibly a morphism.
Let us be more precise in the definition and notation for the 
standard derived-term automaton.

\subsubsection{Definitions and notation}

Let~$\Fd$ be a $\K$-expression over~$M$.
As we have seen, the standard automaton~$\TermF$ has 
dimension~$\DerTerF$ and we write:
\begin{equation}
\TermF = \StanBloc{J}{F}{x}{U}
\eqpnt 
\label{q.trm-F}
\end{equation}

The mere equation~\equnm{trm-F} implies that the scalar~$x$, the 
vectors~$J$ and~$U$ of dimension~$\DerTerF$, as well as the matrix~$F$ 
of dimension~$\DerTerF\x\DerTerF$, are also \emph{associated 
with}~$\Fd$ even though it does appear explicitely in the writing.
When we need to make it more explicit, we write
\begin{equation}
J = \Init{\Fd}\EqVrgInt
F = \Matr{\Fd}\EqVrgInt \e\text{and}\e
U = \Ermi{\Fd}
\eqpnt
\label{q.trm-F-exp}
\end{equation}

By~\equnm{sta-aut-beh} and \propo{cst-trm}, the scalar~$x$ is the 
constant term of~$\Fd$.
By convention, we consider that the vectors~$\Init{\Fd}$
and~$\Ermi{\Fd}$ of dimension~$\DerTerF$ are also of
dimension~$\mathrm{D}$, for any finite
$\msp \mathrm{D}\subset\KRatEM\msp$
that contains~$\DerTerF$, or that~$\Matr{\Fd}$ is a matrix of 
dimension~$\mathrm{D}\x\mathrm{D}$, the `missing' entries being set 
to~$\zeK$.

\subsubsection{The running claims and the preparatory lemmas}

In order to be able to establish Propositions~\ref{p.add-trm} 
and~\ref{p.pro-trm}, that is, to settle the cases of addition and 
product operators, we have to maintain properties, `the claims', all 
along the inductive process, hence for \emph{all} operators.
To this end, we also introduce another function~$\NitlF$, which is 
a vector of dimension~$\DerTerF$ defined inductively as follow.

\medskip

{\allowdisplaybreaks
\noindent
\textbf{Base cases}
\begin{alignat}{3}
\pointn&\e&
\Ed &= \zed \e\text{or}\e \Ed = \und &\ee
\NitlE &=\emptyset \ee \text{vector of dimension~$0$.}
\hskip6.1em
\label{q.ntl-0-1}
\\
\pointn& &
\Ed &= m \e m\in M\bk\unM\e  &
\NitlE  &= (m) \eqpnt
\label{q.ntl-m}
\end{alignat}
\textbf{Dimension invariant operators}
\begin{alignat}{3}
\pointn& &
\Ed &= k\xmd \Fd &\eee
\NitlE &= k\xmd \NitlF \eqpnt
\hskip16.7em
\label{q.ntl-k-lft}
\\
\pointn& &
\Ed &= \Fd\xmd k &
\NitlE &= \NitlF \eqvrg
\label{q.ntl-k-rgt}
\\
 & & 
&\text{\textit{more precisely}} &
\NitlE_{\Kd\xmd k} &= \NitlF_{\Kd}\quantsp
\forall \Kd\in\DerTerF \eqpnt
\notag
\\[1ex]
\pointn&\e&
\Ed &= \Fd^{*} &
\NitlE &= (\TermCnst{\Fd})^{*}\xmd\NitlF \eqvrg
\label{q.ntl-str}
\\
 & & 
&\text{\textit{more precisely}} &
\NitlE_{\Kd\xmd \Fd^{*}} &= (\TermCnst{\Fd})^{*}\xmd\NitlF_{\Kd}\quantsp
\forall \Kd\in\DerTerF \eqpnt
\notag
\end{alignat}
\textbf{Addition and product}
\begin{alignat}{3}
\pointn&\e&
\Ed &= \Fd + \Gd &\ee
\NitlE &= \NitlF + \NitlG \eqvrg
\hskip17em
\label{q.ntl-sum}
\\
 & & 
&\eee \text{\textit{i.e.}} &
\NitlE_{\Kd} &= \NitlF_{\Kd} + \NitlG_{\Kd}\quantsp
\forall \Kd\in\DerTerF\cup\DerTerG \eqpnt
\notag
\\[3ex]
\pointn& &
\Ed &= \Fd \matmul \Gd &
\NitlE &= \NitlF + \TermCnst{\Fd}\xmd\NitlG \eqvrg
\label{q.ntl-prd}
\\
 & &
&\eee\text{\textit{i.e.}}& 
\NitlE_{\Kd} &= \NitlF_{\Kd} + \TermCnst{\Fd}\xmd\NitlG_{\Kd}\quantsp
\forall \Kd\in\DerTerF\matmul\Gd\cup\DerTerG \eqpnt
\notag
\end{alignat}
}

The construction of~$\TermE$ goes with the verification, at every 
step of the induction, of the following properties.

\begin{claim}
\label{cl.int-vec}
$\msp J = \Init{\Ed} = \NitlE\msp$, that is, 
$\msp \forall \Kd\in\DerTerE\quantsmsp J_{\Kd} = \NitlE_{\Kd}\msp$.
\end{claim}

\begin{claim}
\label{cl.trm-cst}
$\msp U = \Ermi{\Ed} = \TermCnst{\DerTerE}\msp$, that is, 
$\msp \forall \Kd\in\DerTerE\quantsmsp U_{\Kd} = \TermCnst{\Kd}\msp$.
\end{claim}

\begin{claim}
\label{cl.row-mat}
For any~$\Kd$ in~$\DerTerE$, the \emph{row} of index~$\Kd$ 
of~$F=\Matr{\Ed}$ is equal to~$\Nitl{\Kd}$, that is,
\begin{equation}
\forall \Kd,\Hd \in \DerTerE \quantsp
\Matr{\Ed}_{\Kd,\Hd} = \Nitl{K}_{\Hd}
\eqpnt 
\label{q.cla-3}
\end{equation}
\end{claim}

The idea behind the definition of~$\NitlE$ and the claims is that we
have, at every step of the induction, the knowledge on~$\TermE$
necessary to prove that the maps~$\varphi$ or~$\psi$ are morphisms
when the operators addition or product come into play.
Before getting to the induction itself, we state some preparatory 
lemmas.

\begin{lemma}
\label{l.dtr-inc}
Let~$\Ed$ be a $\K$-expression over~$M$.
If $\msp\Kd\in\DerTerE\msp$, then $\msp\DerTer{\Kd}\subseteq\DerTerE$.
\end{lemma}

\lemme{dtr-inc} will be used under the following form.

\begin{lemma}
\label{l.dtr-inc-str}
Let~$\Fd$ be a $\K$-expression over~$M$.
If $\msp\Hd\in\DerTerF\msp$, then 
$\msp\DerTer{\Hd\matmul\Fd^{*}}\subseteq\DerTer{\Fd^{*}}$.
\end{lemma}

\subsubsection{The induction: the base cases}

\pointn
$\msp \Ed = \zed \msp$
then
$\msp 
\begin{displaystyle}
\Term{\zed} = \aut{\redmatu{\matriceuu{1}},
                   \redmatu{\matriceuu{0}},
                   \redmatu{\matriceuu{0}}} 
\end{displaystyle}
\msp$.

\pointn
$\msp \Ed = \zed \msp$
then
$\msp 
\begin{displaystyle}
\Term{\und} = \aut{\redmatu{\matriceuu{1}},
                   \redmatu{\matriceuu{0}},
                   \redmatu{\matriceuu{1}}} 
\end{displaystyle}
\msp$.

In both cases, Claims~1, 2, and~3 are obvious by the emptyness 
of~$\DerTer{\zed}$ and~$\DerTer{\und}$.

\pointn
$\msp \Ed = m \in M \msp$
then
$\msp 
\begin{displaystyle}
\Term{m} = 
\aut{\redmatu{\ligned{1}{0}},
          \redmatu{\matricedd{0}{m}{0}{0}},
          \redmatu{\vecteurd{0}{1}}} 
\end{displaystyle}
\msp$.

\clmcl{int-vec} holds by~\equnm{ntl-m}, \clmcl{trm-cst} 
since~$\DerTer{m}=\und$.
Since~$\DerTer{\und}$ is empty, it follows from our convention 
that~$\Nitl{\und}$ is the null vector of any dimension and we have 
here 
$\msp \Matr{\und}_{\und,\und} = 0 = \Nitl{\und}_{\und}\msp$.

\subsubsection{The induction: the dimension invariant operations}

\pointn
$\msp \Ed = k\xmd\Fd \msp$
then
$\msp 
\begin{displaystyle}
\TermE = k\xmd\TermF = \StanBloc{k\xmd J}{F}{k\xmd x}{U} 
\end{displaystyle}
\msp$.\\

\noindent
\clmcl{int-vec} holds by~\equnm{ntl-k-rgt}.
Since \clmcl{trm-cst} and \clmcl{row-mat} hold for~$\TermF$, they also
hold for~$\TermE$ as~$\DerTerE=\DerTerF$, $\Ermi{\Ed}=\Ermi{\Fd}$,
and~$\Matr{\Ed}=\Matr{\Fd}$.\\

\pointn
$\msp \Ed = \Fd\xmd k \msp$
then
$\msp 
\begin{displaystyle}
\TermE = \TermF\xmd k = \StanBloc{J}{F}{x\xmd k}{U\xmd k} 
\end{displaystyle}
\msp$.\\

\noindent
\clmcl{int-vec} holds by~\equnm{ntl-k-lft}.
\clmcl{trm-cst} follows from:
\begin{equation}
\forall \Kd\in\DerTerE \quantsmsp
\Kd=\Hd\xmd k 
\e\text{with}\e
\Hd\in\DerTerF \quantsp
\Ermi{\Ed}_{K} = \Ermi{\Fd}_{H}\xmd k 
               = \TermCnst{\Hd}\xmd k = \TermCnst{\Kd}
\eqpnt
\notag
\end{equation}
\clmcl{row-mat} follows from the fact that by~\equnm{ntl-k-rgt}
$\msp \Nitl{\Hd\xmd k} = \Nitl{\Hd}\msp$ and by~\equnm{sta-aut-rml}, 
and the adequate renaming of row- and column-indices,
$\msp \Matr{\Fd\xmd k} = \Matr{\Fd}\msp$.\\

\pointn
$\msp \Ed = \Fd^{*} \msp$
then
\begin{equation}
\TermE = \left(\TermF\right)^{*} = \StanBloc{x^{*} J}{H}{x^{*}}{U\xmd x^{*}} 
\eqvrg
\notag
\end{equation}
with 
$\msp H = U\matmul x^{*}J + F\msp$.

\noindent
\clmcl{int-vec} holds by~\equnm{ntl-str}.
\clmcl{trm-cst} follows from:
\begin{equation}
\forall \Kd\in\DerTerE\e
\Kd=\Hd\matmul\Fd^{*} 
\e\text{with}\e
\Hd\in\DerTerF 
\e\text{and}\e
\Ermi{\Ed}_{\Kd} = \Ermi{\Fd}_{\Hd}\xmd x^{*} 
           = \TermCnst{\Hd}\xmd\TermCnst{\Fd^{*}} = \TermCnst{\Kd}
\eqpnt
\notag
\end{equation}
In order to prove \clmcl{row-mat}, let~$\Kd$ in~$\DerTerE$, 
hence~$\Kd=\Hd\matmul\Fd^{*}$ with~$\Hd$ in~$\DerTerF$.
By~\equnm{ntl-prd},
\begin{equation}
\Nitl{\Kd} = \Nitl{\Hd} + \TermCnst{\Hd}\xmd\Nitl{\Fd^{*}}
           = \Nitl{\Hd} + \TermCnst{\Hd}\xmd x^{*}\xmd\Nitl{\Fd}
\eqpnt 
\label{q.prf-str}
\notag
\end{equation}
By \lemme{dtr-inc-str}, 
$\msp\DerTer{\Kd}\subseteq\DerTer{\Fd^{*}}\msp$, that is the 
dimension of~$\Nitl{\Kd}$ is contained in~$\DerTerE$.
By induction the claims imply that~$\Nitl{\Hd}$ is the row of 
index~$\Hd$ of~$\Matr{\Fd}=F$ and
$\msp\TermCnst{\Hd}=\Ermi{\Fd}_{\Hd}=U_{\Hd}\msp$.
Hence~\equnm{prf-str} tells that~$\Nitl{\Kd}$ is equal to the row of 
index~$\Kd$ of~$\Matr{\Fd^{*}}=H$.

\subsubsection{The operations addition and product}

In addition to the notation taken in~\equnm{trm-F}, let~$\Gd$ be 
another $\K$-expression and
\begin{equation}
\TermG = \StanBloc{K}{G}{y}{V}
\notag
\end{equation}
its standard term automaton, which fulfil the running claims.


\pointn
$\msp \Ed = \Fd + \Gd \msp$.
We first form
\begin{equation}
\TermF + \TermG = 
\aut{\redmatu{\lignetblblvs{1}{0}{0}},
     \redmatu{\matricettblblvs{0}{J}{K}%
                              {0}{F}{0}%
                              {0}{0}{G}},
     \redmatu{\vecteurtblblvs{x+y}{U}{V}}}\eqvrg
\notag
\end{equation}
a standard automaton of dimension
$\msp\DerTerF\sqcup\DerTerG\msp$.
Let~$\varphi$ be the map
\begin{equation}
\varphi \colon \DerTerF\sqcup\DerTerG \rightarrow
              \DerTerF\cup\DerTerG = \DerTerE
\eqvrg
\notag
\end{equation}
that maps two terms~$\Kd$ and~$\Kd'$ 
of~$\DerTerF\sqcup\DerTerG$ onto one if they are \emph{equal}, hence 
if
$\msp\Kd\in\DerTerF\msp$,
$\msp\Kd'\in\DerTerG\msp$ and 
$\msp\Kd=\Kd'\msp$, or, to state it otherwise, if
$\msp\Kd\in\DerTerF\cap\DerTerG\msp$.

By \clmcl{trm-cst} and \clmcl{row-mat} for~$\TermF$ and~$\TermG$, if 
$\msp\Kd\in\DerTerF\cap\DerTerG\msp$ then
$\msp U_{\Kd} = V_{\Kd} \msp$ and 
$\msp F_{\Kd,.} = G_{\Kd,.}=\Nitl{\Kd}\msp$. 
This is sufficient for~$\varphi$ to be a morphism of automata and 
establishes \propo{add-trm}.

With our convention, both~$J$ and~$K$ can be considered as vectors of 
dimension~$\DerTerF\cup\DerTerG=\DerTerE$.
The image
$\msp\varphi(\TermF+\TermG)\msp$ is~$\TermE$ and can be written
\begin{equation}
\TermE = \StanBloc{J+K}{H}{x+y}{W}
\notag
\end{equation}
where~$H$ is the `fusion' of~$F$ and~$G$ and~$W$ is the `fusion' 
of~$U$ and~$V$.

\clmcl{int-vec} then holds by~\equnm{ntl-sum}.
\clmcl{trm-cst} and \clmcl{row-mat} are directly inherited from the 
corresponding properties for~$\TermF$ and~$\TermG$ (and the 
convention).\\

\pointn
$\msp \Ed = \Fd \matmul \Gd \msp$.
We first form
\begin{equation}
\TermF \matmul \TermG = 
\aut{\redmatu{\lignetblblvs{1}{0}{0}},
     \redmatu{\matricettblblvs{0}{J}{x\xmd K}%
                              {0}{F}{U\matmul K}%
                              {0}{0}{G}},
     \redmatu{\vecteurtblblvs{x\xmd y}{U\xmd y}{V}}}\eqvrg
\notag
\end{equation}
a standard automaton \textit{a priori} of dimension
$\msp\DerTerF\sqcup\DerTerG\msp$, but which we consider as a standard 
automaton of dimension
$\msp\DerTerF\matmul\Gd\sqcup\DerTerG\msp$, that is, we multiply all 
indices from~$\DerTerF$ by~$\Gd$ on the right.

Let~$\psi$ be the `natural' map
\begin{equation}
\psi \colon \DerTerF\matmul\Gd\sqcup\DerTerG \rightarrow
              \DerTerF\matmul\Gd\cup\DerTerG = \DerTerE
\eqvrg
\notag
\end{equation}
that is, $\psi$ maps two terms~$\Kd$ and~$\Kd'$ 
of~$\DerTerF\matmul\Gd\sqcup\DerTerG$ onto one if they are 
\emph{equal}, hence if 
$\msp\Kd\in\DerTerF\matmul\Gd\msp$,
$\msp\Kd'\in\DerTerG\msp$ and 
$\msp\Kd=\Kd'\msp$, or, to state it otherwise, if
$\msp\Kd\in\DerTerF\matmul\Gd\cap\DerTerG\msp$.

Let~$\Kd$ be such an expression, that is,
$\Kd=\Hd\matmul\Gd$ with~$\Hd$ in~$\DerTerF$ and~$\Kd$ belongs 
to~$\DerTerG$.
We consider first the final vector of~$\TermE$.

By \clmcl{trm-cst}, we have on one hand
$\msp\Ermi{\Gd}_{\Kd}=V_{\Kd}=\TermCnst{\Kd}\msp$
and on the other hand
$\msp\TermCnst{\Kd}=\TermCnst{\Hd}\xmd\TermCnst{\Gd}
                   =U_{\Hd}\xmd y\msp$.
Hence, the two entries of index~$\Kd$ of~$U\xmd y$ and~$V$ are equal. 

We then consider~$\Matr{\Ed}$.
By \clmcl{row-mat}, the row of index~$\Kd$ in~$G$ is~$\Nitl{\Kd}$ 
which, by~\equnm{ntl-prd}, is written as
$\msp\Nitl{\Kd} = \Nitl{\Hd} + \TermCnst{\Hd}\xmd\Nitl{\Gd}\msp$, 
that is,
\begin{equation}
  \forall \Ld\in\DerTerG \quantsp
\Nitl{\Kd}_{\Ld} = \Nitl{\Hd}_{\Ld}+\TermCnst{\Hd}\xmd\Nitl{\Gd}_{\Ld}
\eqvrg 
\notag
\end{equation}
which implies in particular that the non-zero entries of~$\Nitl{\Hd}$ 
all correspond to derived terms of~$\Gd$.

The same \clmcl{row-mat} on the other hand implies that the row of 
index~$\Hd$ of~$F$ (of index~$\Hd\matmul\Gd$ in~$\TermE$) 
is~$\Nitl{\Hd}$.
The row of index~$\Hd$ of the matrix~$U\matmul K$ 
is~$\TermCnst{\Hd}\xmd\Nitl{\Gd}$.
If we sum all entries of equal index in~$\DerTerF\matmul\Gd$ on one 
hand and in~$\DerTerG$ on the other hand, we obtain a row-vector~$Z$ 
such that
\begin{equation}
  \forall \Ld\in\DerTerF\matmul\Gd\cup\DerTerG \quantsp
Z_{\Ld} = \Nitl{\Hd}_{\Ld}+\TermCnst{\Hd}\xmd\Nitl{\Gd}_{\Ld}
\eqvrg 
\notag
\end{equation}
and, with our convention,~$\msp Z=\Nitl{\Kd}\msp$.

Together with the property shown above for~$V$ and~$U\xmd y$ this 
proves that~$\psi$ is a morphism of automata and \propo{pro-trm} is 
established.

Moreover, the same computations establish Claims~1 to~3 
for~$\psi(\Term{\Fd\matmul\Gd})$ and by this fact, complete the 
definition of the standard derived-term automaton.

\subsection{The derived-term automaton}
\label{s.der-ter-aut}

Finally, let us define yet another automaton associated with an 
expression~$\Ed$ which is indeed the one we are ultimately aiming at. 
By convention, we consider that the initial state of~$\TermE$ is 
indexed by~$\Ed$.
If~$\Ed$ belongs to~$\DerTerE$, let~$\omega$ be the `natural' map
\begin{equation}
\omega\colon\Ed\sqcup\DerTerE \rightarrow\DerTerE
\eqpnt
\notag
\end{equation}

\begin{proposition}
\label{p.sta-trm-trm}%
The map~$\omega$ is a morphism of automata.
\end{proposition}

\begin{proof}
Let
\begin{equation}
\TermE = \StanBloc{J}{F}{x}{U}
\eqpnt
\notag
\end{equation}
By \clmcl{int-vec}, $\msp J = \Nitl{\Ed}\msp$.
If~$\Ed$ is in~$\DerTerE$, then, by \clmcl{trm-cst},
$\msp U_{\Ed} = \TermCnst{\Ed} = x\msp$
and, by \clmcl{row-mat},
$\msp F_{\Ed,.} = \Nitl{\Ed}\msp$.
These equalities tell that~$\omega$ is a morphism of automata.
\end{proof}

\renewcommand{\Anti}[1]{\Dc_{#1}}

\begin{definition}
\label{d.ant-trm}%
For every valid $\K$-rational expression~$\Ed$, 
\emph{the derived-term automaton}~$\AntiE$ of~$\Ed$ is defined by
$\msp\AntiE=\omega(\TermE)\msp$ if~$\Ed$ is in~$\DerTerE$ and
$\msp\AntiE=\TermE\msp$ otherwise.
\end{definition}

We then finally can state:

\begin{theorem}
\label{t.ant-trm}%
For every valid $\K$-rational expression~$\Ed$, the
derived-term automaton~$\AntiE$ is a quotient of the standard 
automaton of~$\Ed$,~$\StanE$ (and hence realises the 
series denoted by~$\Ed$). 
\EoP
\end{theorem}

\setcounter{theor-tmp}{\value{theorem}}
\setcounter{theorem}{\value{exa-exp}}
\begin{example}[Continued]
Let $\msp\Ed_{\idex}= a^*\matmul(a^* + (-1)b^*)^*\msp$.\\
We have seen that: 
$\msp\DerTer{\Ed_{\idex}}= 
\{a^*\matmul(a^* + (-1)b^*)^*,\msp b^*\matmul(a^* + (-1)b^*)^*\}\msp$
and
\begin{equation}
\Term{\Ed_{\idex}}=
\aut{\lignet{1}{0}{0},
     \matricett{0}{2a}{-b}{0}{2a}{-b}{0}{a}{0},
	 \vecteurt{1}{1}{1}}
\eqpnt
\notag
\end{equation}

It holds that~$\Ed_{\idex}$ is in~$\DerTer{\Ed_{\idex}}$ and we 
observe that the first and second lines of the matrix, as well as the 
first and second entries of the final vector, both indexed by 
instances of the derived term~$\Ed_{\idex}$, are equal.
The quotient of~$\Term{\Ed_{\idex}}$ by the morphism~$\omega$ is:
\\\centering 
$\msp
\begin{displaystyle}
\Anti{\Ed_{\idex}}=
\aut{\ligned{1}{0},
     \matricedd{2a}{-b}{a}{0},
	 \vecteurd{1}{1}}
\end{displaystyle}
\msp$
\ee drawn as \ee
\VCDraw{%
\begin{VCPicture}{(-2,-2)(5,2)}
    \State{(0,0)}{A}\State{(3,0)}{B}
    \Initial[w]{A}%
    \Final[s]{A}\Final[s]{B}%
    \ArcL{A}{B}{-b}\ArcL{B}{A}{a}%
	\LoopN{A}{2\xmd a}%
\end{VCPicture}%
}%

The expression~$\Ed_{\idex}$ has also the property that the `Thompson 
construction' (when generalised to weighted automata) applied to it 
yields a \emph{non-valid} automaton (see~\cite{LombSaka13}).
\end{example}
\setcounter{theorem}{\value{theor-tmp}}


\section{{Back to derivation}}%
\label{s.bac-der}%

Finally, we reconnect this work with the previous ones and show 
that the derived-term automaton we have just described coincides --- in 
the case where~$M$ is a free monoid --- with the automaton defined by 
the derivation of expressions process introduced in~\cite{Anti96} for 
Boolean automata and in~\cite{LombSaka05a} for weighted automata (see 
also~\cite{Saka09,Saka09b,Saka21}).

\subsection{Preparation: the differential of an expression}%
\label{s.dif-exp}%

We begin with a definition and a property that are valid in the case 
of general (graded) monoids.
The specialisation to the case of free monoids allows a particular 
writing that will be used in the sequel.

\begin{definition}
\label{d.dif-exp}%
Let~$\Ed$ be a $\K$-expression over~$M$.
The \emph{differential of~$\Ed$}, denoted by~$\DiffE$, is the 
expression 
\begin{equation}
\DiffE = \sum_{\Hd\in\DerTerE} \NitlE_{\Hd}\matmul\Hd
\eqpnt
\label{q.dif-exp}
\end{equation}
\end{definition}

\equat{dif-exp} allows to write a `first-order development' of the 
expression \textit{via} the following statement.

\begin{proposition}
\label{p.beh-dif}%
\ee
$\msp\displaystyle{%
\CompExpr{\Ed} = \TermCnst{\Ed} + \CompExpr{\DiffE}}\msp$.
\end{proposition}

\begin{proofss}
By induction on the formation of~$\Ed$.
\propo{beh-dif}, 
which we rather write under the form
$\msp
\CompExpr{\Ed} = \TermCnst{\Ed} + 
   \sum_{\Hd\in\DerTerE} \NitlE_{\Hd}\xmd\CompExpr{\Hd}
\msp$,
is based on Equations~(\ref{q.ntl-0-1}) 
to~(\ref{q.ntl-prd}) which have been established with the 
construction of the standard derived-term automaton.

\medskip
\noindent
\textbf{Base cases}
\smallskip

\pointn
$\msp \Ed = \zed \msp$ and $\msp \Ed = \zed \msp$
\eee
obvious by the emptyness of~$\DerTerE$.

\pointn
$\msp \Ed = m \in M \msp$
\ee 
as obvious since
$\msp\CompExpr{m} = m\msp$,
$\msp\TermCnst{m} = \zeK\msp$,
$\msp\DerTer{m} = \und\msp$ and
$\msp\Nitl{m}_{\und} = m\msp$.

\medskip
\noindent
\textbf{Induction}
\smallskip

\pointn
$\msp \Ed = k\xmd\Fd \msp$
\ee
$\msp\TermCnst{k\xmd\Fd} =k\xmd\TermCnst{\Fd}\msp$,
$\msp\DerTer{k\xmd\Fd}=\DerTerF\msp$ and
$\msp\Nitl{k\xmd\Fd} =k\xmd\Nitl{\Fd}\msp$,\\
\eee hence \ee
$\msp\TermCnst{\Ed}+\CompExpr{\DiffE}=
k\xmd\TermCnst{\Fd}+k\xmd
\sum_{\Hd\in\DerTerF} \NitlF_{\Hd}\xmd\CompExpr{\Hd}~=
k~\CompExpr{\Fd}=\CompExpr{\Ed}\msp$.\\[1ex]

\pointn
$\msp \Ed = \Fd\xmd k \msp$
\ee
$\msp\TermCnst{\Fd\xmd k} =\TermCnst{\Fd}\xmd k\msp$,
$\msp\DerTer{\Fd\xmd k}=\DerTerF\xmd k\msp$ and
$\msp\Nitl{\Fd\xmd k} = \Nitl{\Fd}\msp$,\\
\eee hence \ee
$\msp\TermCnst{\Ed}+\CompExpr{\DiffE}=
\TermCnst{\Fd}\xmd k+
\sum_{\Hd\in\DerTerF} \NitlF_{\Hd}\xmd\CompExpr{\Hd\xmd k}~=
\CompExpr{\Fd}~k=\CompExpr{\Ed}\msp$.\\[1ex]

\pointn
$\msp \Ed = \Fd\plusexp\Gd \msp$.
\ee
$\msp\TermCnst{\Fd\plusexp\Gd} =\TermCnst{\Fd}+\TermCnst{\Gd}\msp$,
$\msp\DerTer{\Fd\plusexp\Gd}=\DerTerF\cup\DerTerG\msp$\\
\eee\eee and
$\msp\Nitl{\Fd\plusexp\Gd} = \Nitl{\Fd}+\Nitl{\Gd}\msp$, hence \\
\eee  
$\msp\TermCnst{\Ed}+\CompExpr{\DiffE}=
\TermCnst{\Fd}+\TermCnst{\Gd}+
\sum_{\Hd\in\DerTerF\cup\DerTerF} 
    \left(\Nitl{\Fd}+\Nitl{\Gd}\right)\xmd\CompExpr{\Hd}~=
\CompExpr{\Fd}+\CompExpr{\Fd}=\CompExpr{\Ed}\msp$.\\

\pointn
$\msp \Ed = \Fd \prodexp \Gd \msp$.
\ee
$\msp\TermCnst{\Fd\prodexp\Gd} =\TermCnst{\Fd}\xmd\TermCnst{\Gd}\msp$,
$\msp\DerTer{\Fd\prodexp\Gd}=\DerTerF\prodexp\Gd\cup\DerTerG\msp$\\
\eee\eee and
$\msp\Nitl{\Fd\prodexp\Gd}=\Nitl{\Fd}+\TermCnst{\Fd}\xmd\Nitl{\Gd}\msp$,
more precisely:\\
\ee
$\msp\forall\Hd\in\DerTerF\quantsmsp
   \Nitl{\Fd\prodexp\Gd}_{\Hd\prodexp\Gd} =
   \Nitl{\Fd}_{\Hd}+\TermCnst{\Fd}\xmd\Nitl{\Gd}_{\Hd\prodexp\Gd}\msp$
\e and\\
\ee
$\msp\forall\Kd\in\DerTerG\bk\DerTerF\prodexp\Gd\quantsmsp
   \Nitl{\Fd\prodexp\Gd}_{\Kd}=\TermCnst{\Fd}\xmd\Nitl{\Gd}_{\Kd}\msp$.
\e It then comes
\begin{align}
\TermCnst{\Ed}+\CompExpr{\DiffE} &=
  \TermCnst{\Fd}\xmd\TermCnst{\Gd} + 
  \sum_{\Kd\in\DerTerE} \NitlE_{\Kd}\xmd\CompExpr{\Kd}
\notag
\\
&= \TermCnst{\Fd}\xmd\TermCnst{\Gd}
  + \sum_{\Hd\in\DerTerF} \NitlF_{\Hd}\xmd\CompExpr{\Hd\prodexp\Gd}~
  + \TermCnst{\Fd}\xmd \sum_{\Hd\in\DerTerF}
        \NitlG_{\Hd\prodexp\Gd}\xmd\CompExpr{\Hd\prodexp\Gd}
\eee
\notag
\\
& \eee\eee\eee\eee
  + \TermCnst{\Fd}\xmd \sum_{\Kd\in\DerTerG\bk\DerTerF\prodexp\Gd}
        \NitlG_{\Kd}\xmd\CompExpr{\Kd}
\notag
\\
&= \TermCnst{\Fd}\xmd\TermCnst{\Gd}
   + \left(\sum_{\Hd\in\DerTerF}\NitlF_{\Hd}\xmd\CompExpr{\Hd}\right)~
                                                \CompExpr{\Gd}
  +\TermCnst{\Fd}\xmd\sum_{\Kd\in\DerTerG}\NitlG_{\Kd}\xmd\CompExpr{\Kd}
\notag
\\
&= \TermCnst{\Fd}\xmd\left(\TermCnst{\Gd} +
        \sum_{\Kd\in\DerTerG}\NitlG_{\Kd}\xmd\CompExpr{\Kd}\right) 
   + 
   \left(\sum_{\Hd\in\DerTerF}\NitlF_{\Hd}\xmd\CompExpr{\Hd}\right)\xmd
                                                \CompExpr{\Gd}
\notag
\\
 &= \CompExpr{\Fd}~\CompExpr{\Gd}~=~\CompExpr{\Ed}
\eqpnt
\notag
\end{align} 
\medskip

\pointn
$\msp \Ed = \Fd^{*} \msp$
\ee
$\msp\TermCnst{\Fd^{*}} =\left(\TermCnst{\Fd}\right)^{*}\msp$,
$\msp\DerTer{\Fd^{*}}=\DerTerF\prodexp\Fd^{*}\msp$ and
$\msp\Nitl{\Fd^{*}} = \left(\TermCnst{\Fd}\right)^{*}\Nitl{\Fd}\msp$,
hence
\begin{align}
\TermCnst{\Ed}+\CompExpr{\DiffE} &=
\left(\TermCnst{\Fd}\right)^{*} +
\left(\TermCnst{\Fd}\right)^{*} 
  \sum_{\Kd\in\DerTerF}\NitlF_{\Kd}\xmd\CompExpr{\Kd\prodexp\Fd^{*}}
\notag
\\
&=
\left(\TermCnst{\Fd}\right)^{*} +
\left(\TermCnst{\Fd}\right)^{*} 
  \left(\sum_{\Kd\in\DerTerF}\NitlF_{\Kd}\xmd\CompExpr{\Kd}\right)
  ~\CompExpr{\Fd^{*}}
\eqpnt
\notag
\\
\intertext{The term
$\msp\sum_{\Kd\in\DerTerF}\NitlF_{\Kd}\xmd\CompExpr{\Kd}\msp$
is the \emph{proper part}~$\PartProp{\CompExpr{\Fd}\,}$ 
of~$\CompExpr{\Fd}~$.
Let us write $\msp x=\TermCnst{\Fd}\msp$.
It then comes
}
\TermCnst{\Ed}+\CompExpr{\DiffE} &=
x^{*} + x^{*} \xmd \PartProp{\CompExpr{\Fd}\,}\xmd 
                          \left(\CompExpr{\Fd}\right)^{*} =
x^{*} + x^{*} \xmd \PartProp{\CompExpr{\Fd}\,}\xmd 
        x^{*}\left(\PartProp{\CompExpr{\Fd}\,}\xmd x^{*}\right)^{*} =
x^{*} \left(\unK + \PartProp{\CompExpr{\Fd}\,}\xmd 
    x^{*}\left(\PartProp{\CompExpr{\Fd}\,}\xmd x^{*}\right)^{*}\right)
\notag
\\
&= x^{*} \left(\PartProp{\CompExpr{\Fd}\,}\xmd x^{*}\right)^{*} =
   \left(\CompExpr{\Fd}\right)^{*} = \CompExpr{\Ed}
\eqpnt
\tag*{\qedsymbol}
\end{align} 
\end{proofss}

\medskip
\medskip

\emph{If~$M=A^{*}$ is a free monoid}, every entry of~$\NitlE$ is a linear 
combination of letters in~$A$.
In~\equnm{dif-exp}, we can reorder the terms and see the
vector~$\NitlE$ as the sum of $\jsCard{A}$ $\K$-vectors of
dimension~$\DerTerE$ multiplied by the letters of~$A$:
\begin{equation}
\NitlE = 
\sum_{a\in A}\msp \dervEa\matmul a
\eqvrg
\notag
\end{equation}
and the differential becomes
\begin{equation}
\DiffE = 
\sum_{a\in A}\msp a\matmul\!\!\sum_{\Hd\in\DerTerE} \dervEa_{\Hd}\xmd\Hd
\eqpnt
\label{q.dif-exp-2}
\end{equation}

As recalled in the introduction, the \emph{quotient operation} may be 
defined on languages
\begin{equation}
    \fa L\in\jsPart{\Ae}\quantvrg\fa u\in\Ae\quantsp
     u^{-1}L = \Defi{v\in\Ae}{u\xmd v\in L}
     \eqpnt
     \eee
\notag
\end{equation}
and on series over a free monoid
\begin{equation}
    \fa s\in\KA\quantvrg\fa u\in\Ae\quantsp
	u^{-1}s
	\e\text{is defined by}\e
	\fa v\in\Ae\quantsmsp
     \bra{u^{-1}s,v} = \bra{s,u\xmd v}
     \eqpnt
\notag
\end{equation}
From \propo{beh-dif} and \equnm{dif-exp-2}, directly follows then:
\begin{equation}
a^{-1}\CompExpr{\Ed} = \sum_{\Hd\in\DerTerE} \dervEa_{\Hd}\xmd\Hd
\eqpnt
\notag
\end{equation}

\subsection{The derivation of an expression}%
\label{s.der-exp}%

The result of the \emph{derivation of a (Boolean) expression}, as
defined by Antimirov in~\cite{Anti96} after modification of the
definition of \emph{derivatives} by Brzozowski~\cite{Brzo64}, is a 
\emph{set of expressions}.
The result of the \emph{derivation of a weighted expression}, which 
we have defined in~\cite{LombSaka05a} as a direct generalisation of 
the former, is a \emph{linear combination of (weighted) expressions}.

\begin{definition}
\label{d.der-exp}
Let~$\Ed$ be a $\K$-expression over~$\Ae$ and~$a$ in~$A$.
The \emph{derivation} of~$\Ed$ with respect to~$a$, 
denoted by~${\ExpDer{\Ed}}$, 
is a linear combination of expressions in~$\KRA$,
inductively defined by the following formulas.

\medskip

{\allowdisplaybreaks
\noindent
\textbf{Base cases}
\begin{alignat}{2}
\pointn&\ee&
\ExpDer{\zed} &= \ExpDer{\und} = \zeK \eqpnt
\hskip25.3em
\label{q.der-0-1}
\\[1ex]
\pointn&\ee&
\ExpDer{b} &= 
\left \{
\begin{array}{cl}
\unK & \quad \text{if \quad}  b = a \msp, \\
\zeK & \quad \text{otherwise}.
\end{array} \right .
\label{q.der-b-1}
\end{alignat}
\textbf{Induction}
\begin{alignat}{2}
\pointn&\ee&
\ExpDer{(k\xmd\Fd)}  &= k\xmd\ExpDerF \eqpnt
\hskip24.7em
\label{q.der-k-lft}
\\[1ex]
\pointn&\ee&
\ExpDer{(\Fd\xmd k)} &= \left(\left[\ExpDerF\right]\xmd k\right) \eqpnt
\label{q.der-k-rgt}
\\[1ex]
\pointn&\ee&
\ExpDer{(\Fd\plusexp\Gd)} &= \ExpDerF \plusK \ExpDerG \eqpnt 
\label{q.der-sum}
\\[1ex]
\pointn&\ee&
\ExpDer{(\Fd\prodexp\Gd)} &= 
   \left(\left[\ExpDerF\right] \prodexp \Gd\right)\plusK 
   \TermCst{\Fd}\xmd\ExpDerG \eqpnt
\label{q.der-prd}
\\[1ex]
\pointn&\ee&
\ExpDer{(\Fd^{*})} &= \TermCst{\Fd}^{*}\xmd 
   \left(\left[\ExpDerF\right]\prodexp \Fd^{*}\right) \eqpnt
\label{q.der-str}
\end{alignat}
}
\end{definition}

\subsection{The reconciliation}%
\label{s.rec-onc}%

\begin{theorem}
\label{t.rec-onc}%
Let~$\Ed$ be a $\K$-expression over~$\Ae$ and~$a$ in~$A$.
The derivation of~$\Ed$ with respect to~$a$ is the coefficient of~$a$ 
in~$\DiffE$:
\begin{equation}
\ExpDerE = \sum_{\Hd\in\DerTerE} \dervEa_{\Hd}\xmd\Hd
\eqpnt
\label{q.rec-onc}
\end{equation}
\end{theorem}

A direct consequence of this statement is the fact that derivation is 
the lifting of the quotient of series at the level of expressions.

\begin{corollary}
\label{c.rec-onc}%
\ee
$\msp\displaystyle{\CompExpr{\ExpDerE} = a^{-1}\CompExpr{\Ed}}\msp$.
\end{corollary}

\begin{proof}[Proof of \theor{rec-onc}]
It is less a proof than a mere verification without mystery,
by induction on the formation of~$\Ed$, and based on
Equations~(\ref{q.ntl-0-1}) to~(\ref{q.ntl-prd}).

\medskip

\noindent
\textbf{Base cases}

\medskip

\pointn
$\msp \Ed = \zed \msp$ and $\msp \Ed = \zed \msp$
\eee
obvious by the emptyness of~$\DerTerE$.

\pointn
$\msp \Ed = a \in A \msp$
\ee 
as obvious since
$\msp\DerTer{a} = \und\msp$ and
$\msp\Nitl{a}_{\und} = a\msp$.

\medskip

\noindent
\textbf{Induction}

\medskip

\pointn
$\msp \Ed = k\xmd\Fd \msp$
\ee
$\msp\displaystyle{\ExpDer{(k\xmd\Fd)} = k\xmd\ExpDerF}\msp$ 
on one hand-side,\\[1.4ex]
\eee\ee\e\ 
$\DerTer{k\xmd\Fd}=\DerTerF\msp$ and
$\msp\Nitl{k\xmd\Fd} =k\xmd\Nitl{\Fd}\msp$ 
on the other;\\
\eee\ee\e\ 
if~\equnm{rec-onc} holds for~$\Fd$, it holds for~$k\xmd\Fd$.\\

\pointn
$\msp \Ed = \Fd\xmd k \msp$
\ee
$\msp\displaystyle{\ExpDer{(\Fd\xmd k)} = \ExpDerF\xmd k}\msp$ 
on one hand-side,\\[1.4ex]
\eee\ee\e\ 
$\DerTer{\Fd\xmd k}=\DerTerF\xmd k\msp$ and
$\msp\Nitl{\Fd\xmd k} =\Nitl{\Fd}\msp$ 
on the other;\\
\eee\ee\e\ 
if~\equnm{rec-onc} holds for~$\Fd$, it holds for~$\Fd\xmd k$.\\

\pointn
$\msp \Ed = \Fd\plusexp\Gd \msp$.
\ee
$\msp\displaystyle{\ExpDer{(\Fd\plusexp\Gd)}=\ExpDerF\plusK\ExpDerG}\msp$ 
on one hand-side,\\[1.4ex]
\eee\ee\e\ 
$\DerTer{\Fd\plusexp\Gd}=\DerTerF\cup\DerTerG\msp$ and
$\msp\Nitl{\Fd\plusexp\Gd} = \Nitl{\Fd}+\Nitl{\Gd}\msp$ 
on the other;\\
\eee\ee\e\ 
if~\equnm{rec-onc} holds for~$\Fd$ and~$\Gd$, it holds for~$\Fd\plusexp\Gd$.\\

\pointn
$\msp \Ed = \Fd\prodexp\Gd \msp$.
\ee
$\msp\displaystyle{\ExpDer{(\Fd\prodexp\Gd)}= 
   \left(\left[\ExpDerF\right] \prodexp \Gd\right)\plusK 
   \TermCst{\Fd}\xmd\ExpDerG}\msp$ 
on one hand-side,\\[1.6ex]
\ee\ee 
$\DerTer{\Fd\prodexp\Gd}=\DerTerF\prodexp\Gd\cup\DerTerG\msp$
and
$\msp\Nitl{\Fd\prodexp\Gd}=\Nitl{\Fd}+\TermCnst{\Fd}\xmd\Nitl{\Gd}\msp$ 
on the other;\\[1.4ex]

\ee\ee 
more precisely:\\[1.4ex]
\ee\ee
$\forall\Hd\in\DerTerF\quantsmsp
   \Nitl{\Fd\prodexp\Gd}_{\Hd\prodexp\Gd} =
   \Nitl{\Fd}_{\Hd}+\TermCnst{\Fd}\xmd\Nitl{\Gd}_{\Hd\prodexp\Gd}\msp$
\e and\\
\ee\ee
$\forall\Kd\in\DerTerG\bk\DerTerF\prodexp\Gd\quantsmsp
   \Nitl{\Fd\prodexp\Gd}_{\Kd}=\TermCnst{\Fd}\xmd\Nitl{\Gd}_{\Kd}\msp$;\\
\ee\ee 
if~\equnm{rec-onc} holds for~$\Fd$ and~$\Gd$, it holds for~$\Fd\prodexp\Gd$.\\

\pointn
$\msp \Ed = \Fd^{*} \msp$
\ee
$\msp\displaystyle{\ExpDer{(\Fd^{*})}=\TermCst{\Fd}^{*}\xmd 
   \left(\left[\ExpDerF\right]\prodexp \Fd^{*}\right)}\msp$ 
on one hand-side,\\[1.4ex]
\eee\ee\e\ 
$\DerTer{\Fd^{*}}=\DerTerF\prodexp\Fd^{*}\msp$ and
$\msp\Nitl{\Fd^{*}} = \left(\TermCnst{\Fd}\right)^{*}\Nitl{\Fd}\msp$,
on the other;\\
\eee\ee\e\ 
if~\equnm{rec-onc} holds for~$\Fd$, it holds for~$\Fd^{*}$.
\end{proof}

This conclude the proof that the derived-term automaton we have
defined in this paper coincides with the one that was defined in the
previous work dealing with expressions over the free monoids.

It is noteworthy that other works that dealt with the derivation of
expressions outside from the scope of the free
monoid~\cite{Dema17,KonsEtAl21} have considered entities which are closed
to ours.
In particular, the differential of an expression is called the
\emph{linear form} in~\cite{KonsEtAl21}, and the sum of the
differential and the constant term is the \emph{expansion}
in~\cite{Dema17}.

Nevertheless, we have taken here the formalism to its logical
conclusion and designed a construction of the derived-term automaton
that gets rid of the derivation, derivatives or their analogues.



\addcontentsline{toc}{section}{References}%

{\small
\bibliographystyle{\BibDir eta.en}
\bibliography{\BibDir DTWD-bib}

\def\No{{$^\circ$}}%
\def\showlabel#1{}%
\begin{thebibliography}{10}

\bibitem{Anti96}
{V.~Antimirov}, Partial derivatives of regular expressions and finite automaton
  constructions, {\em Theoret. Computer Sci.}, vol.~155 (1996),
  291--319.\showlabel{Anti96}

\bibitem{BersReut88}
{J.~Berstel and C.~Reutenauer}, {\em Rational Series and Their Languages},
  Springer, 1988.
\newblock Translation of \emph{Les s{\'e}ries rationnelles et leurs langages}
  Masson, 1984.\showlabel{BersReut88}

\bibitem{BersReut11}
{J.~Berstel and C.~Reutenauer}, {\em Noncommutative Rational Series with
  Applications}, Cambridge University Press, 2011.
\newblock New version of \emph{Rational Series and Their Languages.} Springer,
  1988.\showlabel{BersReut11}

\bibitem{Brzo64}
{J.~A. Brzozowski}, Derivatives of regular expressions, {\em J. Assoc. Comput.
  Mach.}, vol.~11 (1964), 481--494.\showlabel{Brzo64}

\bibitem{CaroFlou03}
{P.~Caron and M.~Flouret}, Glushkov Construction for Series: The Non
  Commutative Case, {\em Int. J. Comput. Math.}, vol.~80,4 (2003),
  457--472.\showlabel{CaroFlou03}

\bibitem{ChamEtAl09}
{J.-M. Champarnaud, F.~Ouardi and D.~Ziadi}, An Efficient Computation of the
  Equation {K}-automaton of a Regular {K}-expression, {\em Fundam. Inform.},
  vol.~90,1-2 (2009), 1--16.\showlabel{ChamEtAl09}

\bibitem{ChamZiad02}
{J.-M. Champarnaud and D.~Ziadi}, Canonical derivatives, partial derivatives
  and finite automaton constructions, {\em Theoret. Computer Sci.}, vol.~289
  (2002), 137--163.\showlabel{ChamZiad02}

\bibitem{Dema17}
{A.~Demaille}, Derived-Term Automata of Multitape Expressions with Composition,
  {\em Sci. Ann. Comput. Sci.}, vol.~27,2 (2017), 137--176.\showlabel{Dema17}

\bibitem{DiekRoze95}
{V.~Diekert and G.~Rozenberg} (ed.), {\em The Book of Traces}, World
  Scientific, 1995.\showlabel{DiekRoze95}

\bibitem{DrosEtAl09Ed}
{M.~Droste, W.~Kuich and H.~Vogler} {(Ed.)}, {\em Handbook of {W}eighted
  {A}utomata}, Springer, 2009.\showlabel{DrosEtAl09Ed}

\bibitem{Glus61}
{V.~M. Glushkov}, The abstract theory of automata, {\em Russian {M}ath.
  {S}urveys}, vol.~16 (1961), 1--53.\showlabel{Glus61}

\bibitem{KonsEtAl21}
{S.~Konstantinidis, N.~Moreira and R.~Reis}, Partial derivatives of regular
  expressions over alphabet-invariant and user-defined labels, {\em Theoret.
  Computer Sci.}, vol.~870 (2021), 103--120.\showlabel{KonsEtAl21}

\bibitem{LombSaka05a}
{S.~Lombardy and J.~Sakarovitch}, Derivatives of rational expressions with
  multiplicity, {\em Theoret. Computer Sci.}, vol.~332 (2005),
  141--177.\showlabel{LombSaka05a}

\bibitem{LombSaka13}
{S.~Lombardy and J.~Sakarovitch}, The validity of weighted automata, {\em Int.
  J. of Algebra and Computation}, vol.~23,4 (2013),
  863--914.\showlabel{LombSaka13}

\bibitem{MadoSakaxx}
{D.~Madore and J.~Sakarovitch}, An example of a non strong {B}anach algebra, in
  preparation.\showlabel{MadoSakaxx}

\bibitem{Pin21Ed}
{J.-{\'E}. Pin} {(Ed.)}, {\em {H}andbook of {A}utomata {T}heory, {Vol.~{I}
  and~{II}}}, {E}uropean {M}athematical {S}ociety {P}ress,
  2021.\showlabel{Pin21Ed}

\bibitem{Rutt03}
{J.~M. Rutten}, Behavioural differential equations: a coinductive calculus of
  streams, automata, and power series, {\em Theoret. Computer Sci.}, vol.~308
  (2003), 1--53.\showlabel{Rutt03}

\bibitem{Saka09}
{J.~Sakarovitch}, {\em Elements of Automata Theory}, Cambridge University
  Press, 2009.
\newblock Corrected English translation of \emph{\'El\'ements de th\'eorie des
  automates}, Vuibert, 2003.\showlabel{Saka09}

\bibitem{Saka09b}
{J.~Sakarovitch}, Rational and recognisable power series, {in:} {\em Handbook
  of Weighted Automata}, M.~Droste, W.~Kuich and H.~Vogler (ed.), Springer,
  2009, 105--174.\showlabel{Saka09b}

\bibitem{Saka21}
{J.~Sakarovitch}, Automata and expressions, {in:} {\em {H}andbook of {A}utomata
  {T}heory, Vol.~{I}}, J.-{\'E}. Pin (ed.), European {M}athematical {S}ociety
  {P}ress, 2021, 39--78.\showlabel{Saka21}

\bibitem{SaloSoit77}
{A.~Salomaa and M.~Soittola}, {\em Automata-Theoretic Aspects of Formal Power
  Series}, Springer, 1977.\showlabel{SaloSoit77}

\end{thebibliography}
}


\end{document}